\documentclass[11pt]{article}

\usepackage{anysize} % a simple package to set up document margins
\marginsize{2cm}{2cm}{2cm}{2cm}

\usepackage{amsmath} % enhances the typeset appearance of mathematical formulas
\usepackage{amssymb} % subset of amsfonts that defines the full set of symbol names for two fonts of extra symbols included in amsfonts
\usepackage{amsthm} % provides an enhanced version of LaTeX's \newtheorem command for defining theorem-like environments

\usepackage{booktabs} % publication quality tables in LaTeX
\usepackage{complexity} % computational complexity class names

\usepackage{color} % color control for LaTeX documents
\usepackage{colortbl} % add color to LaTeX tables

\usepackage{enumitem} % control layout of itemize, enumerate, description

\usepackage{algorithmic} % pseudocode support
\usepackage{array} % extending the array and tabular environments

\usepackage[pdftex,pagebackref=true,colorlinks]{hyperref} % extensive support for hypertext in LaTeX
  \hypersetup{linkcolor=blue,filecolor=blue,citecolor=blue,urlcolor=blue}
\usepackage{algorithm} % floating environment for algorithms
\usepackage{tabularx} % tabulars with adjustable-width columns

\usepackage{bbm}

\newcommand{\bra}[2][]{\mathinner{\langle #2|}_{#1}}
\newcommand{\ket}[2][]{\vert#2\rangle_{\hspace{-0.1em}#1}}
\newcommand{\braket}[3][]{\mathinner{\langle #2|#3\rangle}_{\hspace{-0.1em}#1}}
\newcommand{\ketbra}[3][]{\mathinner{|#2\rangle\langle #3|}_{#1}}

\newcommand{\abs}[1]{|#1|}
\newcommand{\Abs}[1]{\left\vert #1\right\vert}

\def\clap#1{\hbox to 0pt{\hss#1\hss}}

\DeclareMathOperator{\tr}{Tr}

\usepackage{xy}
\xyoption{matrix}
\xyoption{frame}
\xyoption{arrow}
\xyoption{arc}

\usepackage{ifpdf}
\ifpdf
\else
\PackageWarningNoLine{Qcircuit}{Qcircuit is loading in Postscript mode.  The Xy-pic options ps and dvips will be loaded.  If you wish to use other Postscript drivers for Xy-pic, you must modify the code in Qcircuit.tex}
\xyoption{ps}
\xyoption{dvips}
\fi

\entrymodifiers={!C\entrybox}

\newcommand{\qw}[1][-1]{\ar @{-} [0,#1]}
\newcommand{\qwx}[1][-1]{\ar @{-} [#1,0]}
\newcommand{\cw}[1][-1]{\ar @{=} [0,#1]}
\newcommand{\gate}[1]{*+<.6em>{#1} \POS ="i","i"+UR;"i"+UL **\dir{-};"i"+DL **\dir{-};"i"+DR **\dir{-};"i"+UR **\dir{-},"i" \qw}
\newcommand{\meter}{*=<1.8em,1.4em>{\xy ="j","j"-<.778em,.322em>;{"j"+<.778em,-.322em> \ellipse ur,_{}},"j"-<0em,.4em>;p+<.5em,.9em> **\dir{-},"j"+<2.2em,2.2em>*{},"j"-<2.2em,2.2em>*{} \endxy} \POS ="i","i"+UR;"i"+UL **\dir{-};"i"+DL **\dir{-};"i"+DR **\dir{-};"i"+UR **\dir{-},"i" \qw}
\newcommand{\control}{*!<0em,.025em>-=-<.2em>{\bullet}}
\newcommand{\ctrl}[1]{\control \qwx[#1] \qw}
\newcommand{\multigate}[2]{*+<1em,.9em>{\hphantom{#2}} \POS [0,0]="i",[0,0].[#1,0]="e",!C *{#2},"e"+UR;"e"+UL **\dir{-};"e"+DL **\dir{-};"e"+DR **\dir{-};"e"+UR **\dir{-},"i" \qw}
\newcommand{\ghost}[1]{*+<1em,.9em>{\hphantom{#1}} \qw}
\newcommand{\lstick}[1]{*!R!<.5em,0em>=<0em>{#1}}
\newcommand{\Qcircuit}{\xymatrix @*=<0em>}
\theoremstyle{plain}
\newtheorem{theorem}{Theorem}[section]
\newtheorem{lemma}[theorem]{Lemma}
\newtheorem{proposition}[theorem]{Proposition}
\newtheorem*{proposition*}{Proposition}
\newtheorem{claim}[theorem]{Claim}

\newtheorem*{utheorem}{Theorem}
\newtheorem{corollary}{Corollary}
\newtheorem{definition}{Definition}[section]
\newtheorem{question}{Question}

\theoremstyle{definition}

\newtheorem{remark}[theorem]{Remark}

\theoremstyle{remark}

\newcommand{\namedref}[2]{\hyperref[#2]{#1~\ref*{#2}}}

\newcommand{\sectionref}[1]{\namedref{Section}{#1}}

\newcommand{\remarkref}[1]{\namedref{Remark}{#1}}

\newcommand{\tableref}[1]{\namedref{Table}{#1}}
\newcommand{\lemmaref}[1]{\namedref{Lemma}{#1}}
\newcommand{\claimref}[1]{\namedref{Claim}{#1}}
\newcommand{\propositionref}[1]{\namedref{Proposition}{#1}}
\newcommand{\corollaryref}[1]{\namedref{Corollary}{#1}}

\newcommand{\algorithmref}[1]{\namedref{Algorithm}{#1}}

\newcommand{\questionref}[1]{\namedref{Question}{#1}}

\newcommand{\parhead}[1]{\medskip \noindent {\bfseries\boldmath\ignorespaces #1.}\hskip 0.9em plus 0.3em minus 0.3em}

\newcommand{\veps}{\varepsilon}

\newcommand{\RAM}{\mathsf{RAM}}
\newcommand{\TM}{\mathsf{TM}}
\newcommand{\MTM}{\mathsf{m}\TM}

\newlang{\TSAT}{3SAT}
\newlang{\TCSP}{2CSP}
\newlang{\CSP}{CSP}
\newlang{\TOFSAT}{(2,4)SAT}

\newcommand{\HILBERT}{\mathcal{H}}

\newcommand{\BELL}{\mathsf{Bell}}
\newcommand{\LOCC}{\mathsf{LOCC}}
\newcommand{\SEP}{\mathsf{SEP}}

\newcommand{\CTRS}{\mathcal{C}}

\newcommand{\VARNORM}[1]{\left\vert #1 \right\vert_{1}}
\newcommand{\TRACENORM}[1]{\left\vert #1 \right\vert_{\tr}}

\newcommand{\DIST}{\mathsf{dstr}}

\newcommand{\NUMP}{\kappa}
\newcommand{\PRL}{\ell}

\newcommand{\REJ}{\mathsf{REJ}}

\newcommand{\SWAP}{\mathsf{SWAP}}
\newcommand{\SWAPTEST}{\textsc{Swap}}
\newcommand{\CONSTEST}{\textsc{Cons}}
\newcommand{\EDGECONSTEST}{\textsc{EdgeCons}}
\newcommand{\COLCONSTEST}{\textsc{ColCons}}

\newcommand{\UNITEST}{\textsc{Unif}}
\newcommand{\CONDUNITEST}{\textsc{CondUnif}}

\newcommand{\BPSI}{\Psi}
\newcommand{\BPHI}{\Phi}
\newcommand{\BCHI}{X}

\DeclareMathOperator{\Var}{Var}

\renewcommand{\ket}[2][]{\vert#2\rangle_{\hspace{-0.1em}#1}}
\renewcommand{\bra}[2][]{\langle #2\vert_{#1}}

\title{Improved Soundness for $\mathrm{QMA}$ with Multiple Provers}
\date{\today}
\author{
Alessandro Chiesa \\ \href{mailto:alexch@csail.mit.edu}{alexch@csail.mit.edu} \\MIT
\and
Michael A.\ Forbes 
\thanks{Supported by NSF Grant CCF-0939370.}
\\  \href{mailto:miforbes@mit.edu}{miforbes@mit.edu} \\ MIT
}

\begin{document}

\maketitle

\begin{abstract}
We present three contributions to the understanding of $\QMA$ with multiple provers:
\begin{itemize}
  \item We give a tight soundness analysis of the protocol of [Blier and Tapp, ICQNM '09], yielding a soundness gap $\Omega(N^{-2})$.  Our improvement is achieved \emph{without} the use of an instance with a constant soundness gap (i.e., without using a ``PCP'').
  \item We give a tight soundness analysis of the protocol of [Chen and Drucker, ArXiV '10], thereby improving their result from a ``monolithic'' protocol where $\Theta(\sqrt{N})$ provers are \emph{needed} in order to have any soundness gap, to a protocol with a \emph{smooth trade-off} between the number of provers $\kappa$ and a soundness gap $\Omega(\kappa^{2}N^{-1})$, as long as $\kappa \in \Omega(\log N)$. (And, when $\kappa \in \Theta(\sqrt{N})$, we recover the original parameters of Chen and Drucker.) 
  \item We make progress towards an open question of [Aaronson et al., ToC '09] about what kinds of $\NP$-complete problems are amenable to ``sublinear'' multiple-prover $\QMA$ protocols, by observing that a \emph{large class} of such examples can easily be derived from results already in the PCP literature --- namely, at least the languages recognized by a non-deterministic $\RAM$s in quasilinear time.
\end{itemize}
\end{abstract}

\newpage

\section{Introduction}
\label{sec:introduction}

The class $\QMA$ is the natural quantum analogue of $\NP$ (or, rather, $\MA$): with the help of a quantum proof (given by the all-powerful ``Merlin''), a quantum polynomial-time verifier (``Arthur'') attempts to decide whether an input string $x$ is in a given language $L$ or not; this class was first studied by Knill \cite{Kni96}, Kitaev \cite{Kit99}, and Watrous \cite{Wat00}. For more details, see the survey of Aharonov and Naveh \cite{AN02}.

Kobayashi et al. \cite{KMY03} first introduced and studied the class $\QMA(\NUMP)$, where Arthur receives $\NUMP \in [2,\poly(N)]$ quantum proofs that are \emph{promised to be unentangled}. While multiple proofs in the classical case do not increase the power of the class (i.e., ``$\NP(\NUMP)=\NP$'' and ``$\MA(\NUMP)=\MA$''), there is some evidence that multiple unentangled proofs in the quantum case are in fact more powerful than one (as currently conjectured): for example, Liu et al. \cite{LCV07} have proposed a problem, pure state $N$-representability, that is known to lie in $\QMA(2)$ but is not known to lie in $\QMA$; also, several works \cite{BT09,Bei10,ABDFS09,CD10,GallNN12,Pereszlenyi12} have proposed multi-prover $\QMA$ protocols for certain $\NP$ languages whose (soundness and proof length) parameters are not known to be achievable with only one prover. (See \tableref{table:summary} for a summary of such results.)

Harrow and Montanaro \cite{HM10} recently answered several open problems regarding the class $\QMA(\NUMP)$, by proving that amplification within $\QMA(\NUMP)$ is possible and that $\QMA(\poly(N))=\QMA(2)$; the ``collapse'' is achieved by giving an analysis of a \emph{product test}, which allows a verifier to use the unentanglement promise of only two registers to ensure that states within a single register are close to a separable state.

Brand\~{a}o et al. \cite[Corollary 4]{BCY10} prove (among other things) that two-prover $\QMA$ protocols where the verifier is restricted to LOCC measurements (i.e., adaptive unentangled measurements) only can be simulated by a \emph{single} prover $\QMA$ protocol, incurring only in a quadratic increase in total proof length. In particular, for example, this implies that a two-prover $\QMA$ protocol for $\TSAT$ with an LOCC verifier and total proof length of $o(\sqrt{N})$ is unlikely to exist (for, otherwise, $\TSAT$ could be solved in deterministic subexponential time). In a related theme, Brand\~{a}o and Harrow~\cite{BrandaoH12} show that the $\widetilde{O}(\sqrt{N})$-prover LOCC-protocol of Chen and Drucker~\cite{CD10} is optimal in the $\poly(N)$-prover regime, under the same hardness assumption for $\TSAT$ as above.

Particularly interesting is the gap between the ``lower bound'' of Brand\~{a}o et al. \cite{BCY10} and the ``upper bound'' results that are known for multi-prover $\QMA$ protocols for certain $\NP$ languages. Specifically, Aaronson et al. \cite{ABDFS09} give a $\widetilde{\Theta}(\sqrt{N})$-prover $\QMA$ protocol for $\TSAT$, with perfect completeness and constant soundness gap, where each prover sends $\Theta(\log N)$ qubits; two improvements, in different directions, on this protocol are known:
\begin{itemize}[nolistsep]
  \item \emph{Reducing the number of provers.} Harrow and Montanaro \cite{HM10}, through their product test, reduce the number of provers of \cite{ABDFS09} to only two, thereby obtaining a two-prover $\QMA$ protocol for $\TSAT$, with perfect completeness and constant soundness gap, where each prover sends $\widetilde{\Theta}(\sqrt{N})$ qubits.
  \item \emph{Avoiding use of the swap test (and any entangling measurement).} Chen and Drucker \cite{CD10} simplify the verifier of \cite{ABDFS09} by avoiding the swap test, thereby making the verifier perform only LOCC (in fact, Bell) measurements; along the way, they also manage to greatly simplify the soundness analysis too. (Also, but less relevant: they (i) use a coloring problem as a starting point instead of a ``structured'' $\SAT$ instance, and (ii) they lose perfect completeness.)
\end{itemize}
However, no result that improves on \emph{both} directions is known; such a result, in light of the lower bound of Brand\~{a}o et al. \cite{BCY10}, would be a tight upper bound (under plausible hardness assumptions).
Thus, the following is an interesting open question:

\begin{question}
\label{question:main}
Does there exist a two-prover $\QMA$ protocol for $\TSAT$, with a constant soundness gap and $\widetilde{O}(\sqrt{N})$ total number of qubits, where the verifier is only allowed to perform LOCC measurements?
\end{question}

\begin{table*}[t]
\centering
\begin{tabular}{@{}llllllll@{}}\toprule
\textbf{paper} & \textbf{language} & \textbf{gap?} & \textbf{provers} & $\frac{\textbf{qubits}}{\textbf{provers}}$ & $c$ & $c-s$ & \textbf{verifier test} \\
\midrule
\cite{BT09}    & $\TCSP(N,M,3)$    & no  & $2$                & $\Theta(\log N)$               & $1$                                            &
$\Omega(N^{-6})$             & $\SWAP\,,\,\BELL$  \\

\cite{Bei10}   & $\TOFSAT(N)$  & yes & $2$                & $\Theta(\log N)$               & $\frac{3}{4} + \frac{\sqrt{2(N-1)}}{6N^{1.5}}$ &
$\Omega(N^{-3-\veps})$     & $\SWAP\,,\,\BELL$  \\

\cite{ABDFS09} & $\TOFSAT(N)$  & yes & $\Theta(\sqrt{N})$ & $\Theta(\log N)$               & $1$                                            &
$\Omega(1)$                   & $\SWAP\,,\,\BELL$  \\

\cite{CD10}    & $\TCSP(N,M,O(1))$ & yes & $\Theta(\sqrt{N})$ & $\Theta(\log N)$               & $1 - e^{-\Omega(\sqrt{N})}$                    &
$\Omega(1)$                   & $\BELL$            \\

\cite{HM10}    & $\TOFSAT(N)$  & yes & $2$                & $\widetilde{\Theta}(\sqrt{N})$ & $1$                                            &
$\Omega(1)$                   & $\SWAP\,,\,\BELL$  \\

this work      & $\TCSP(N,M,O(1))$ & no  & $2$                & $\Theta(\log N)$               & $1$                                            &
$\Omega(N^{-2})$              & $\SWAP\,,\,\BELL$  \\

this work      & $\TCSP(N,M,O(1))$ & yes & $\NUMP$            & $\Theta(\log N)$               & $1 - e^{-\Omega(\NUMP)}$                       &
$\Omega(\NUMP^{2}N^{-1})$ & $\BELL$            \\
\cite{GallNN12}      & $\TCSP(N,M,O(1))$ & yes & $\NUMP$            & $\Theta(\log N)$               & $1$                       &
$\Omega(N^{-1})$ & $\SWAP\,,\,\BELL$            \\

\bottomrule
\end{tabular}
\label{table:summary}
\caption{A summary of the known multi-prover $\QMA$ protocols for languages in $\NP$. The language $\TCSP(N,M,K)$ consists of satisfiable $\TCSP$ instances on $N$ vertices, with $M$ (edge) constraints and an alphabet size of $K$; the language $\TOFSAT(N)$ consists of satisfiable $2$-out-of-$4$ balanced $\SAT$ instances on $N$ variables. The ``gap?'' field indicates whether the language is assumed to have a constant gap in soundness, so that a PCP is required to transfer the results to $\NP$. See \sectionref{sec:preliminaries} for formal definitions of these languages.}
\end{table*}

\subsection{Our Contributions}

In this work, we present contributions that further our understanding (though do not resolve) \questionref{question:main}. Specifically,
\begin{itemize}
  \item We address an open question of Aaronson et al.~\cite{ABDFS09} about what kinds of $\NP$-complete problems are amenable to ``sublinear'' multiple-prover $\QMA$ protocols, by observing that a \emph{large class} of such examples can easily be derived from results in the PCP literature.
  \item We give a tight soundness analysis of the protocol of Chen and Drucker \cite{CD10} for $\TSAT$, thereby improving their result from a ``monolithic'' protocol where $\widetilde{\Theta}(\sqrt{N})$ provers are \emph{needed} in order to have any soundness gap, to a protocol with a \emph{smooth trade-off} between the number of provers $\NUMP$ and a soundness gap $\widetilde{\Omega}(\NUMP^{2}N^{-1})$, as long as $\NUMP \in \Omega(\log N)$. (And, when $\NUMP \in \widetilde{\Theta}(\sqrt{N})$, we recover the original parameters of \cite{CD10}.) Further, we explain why even our tight analysis cannot give any soundness gap for the ``$\NUMP \in O(1)$ regime'', implying that new protocols are needed for any ``sublinear'' constant-prover LOCC $\QMA$ protocol with an inverse-polynomial soundness gap.
  \item We give a tight soundness analysis of the protocol of Blier and Tapp \cite{BT09} for $\TSAT$, yielding a soundness gap $\Omega(N^{-2})$. Maybe surprisingly, our improvement is achieved \emph{without} the use of an instance with a constant soundness gap (i.e., without using a ``PCP''); this is unlike the soundness gap of $\widetilde{\Omega}(N^{-3-\veps})$ given by Beigi \cite{Bei10}, which was achieved using a (balanced) $2$-out-of-$4$ instance with constant soundness gap. Independently from us, Le Gall et al.\ \cite{GallNN12} have been able to use PCPs in the protocol of Blier and Tapp \cite{BT09} to obtain a soundness game of $\widetilde{\Omega}(N^{-1})$.
\end{itemize}
We now discuss each of the above contributions; the technical details are left to subsequent sections.

\subsubsection{Quasilinear-time has sublinear unentangled quantum proofs}

The main theorem of Aaronson et al. \cite{ABDFS09} is:

\begin{utheorem}[{\cite[Theorem 1]{ABDFS09}}]
\label{thm:ABDFS09}
Let $\varphi$ be a $\TSAT$ instance with $N$ variables and $M$ clauses (and $M \geq N$). Then one can prove satisfiability of $\varphi$, with perfect completeness and constant soundness, using $\widetilde{\Theta}(\sqrt{M})$ unentangled quantum proofs, each with $\Theta(\log N)$ qubits.
\end{utheorem}

What is surprising about the result is that, for $M \in O(N)$, the total number of qubits sent by all the provers to the verifier is \emph{sublinear}; instead, the best known proof length in the case of only one prover (i.e., in the case of $\QMA$) is linear (and we believe one cannot do better, by the Exponential-Time Hypothesis \cite{ImpagliazzoR01}, which says that $\TSAT$ cannot be solved in subexponential time in the worst case). Aaronson et al. \cite{ABDFS09} thus raised the following question:
\begin{question}
\label{question:ABDFS09}
For which $\NP$-complete problems does an analogue of \cite[Theorem 1]{ABDFS09} hold?
\end{question}

We note that the existing PCP literature already yields a large class of languages for which an analogue of \cite[Theorem 1]{ABDFS09} does hold. Concretely, the starting point of Aaronson et al. \cite{ABDFS09} was the existence of a quasilinear reduction from $\TSAT$ to $\TCSP$ with \emph{constant} soundness gap; we note that the works of Ben-Sasson and Sudan \cite{BS05} and Dinur \cite{Din07} actually imply that a similar reduction holds for any language that can be recognized in non-deterministic quasilinear time by a random-access machine.

\begin{proposition}
\label{prop:extend-aaronson}
\label{proposition:sublinear}
Let $L$ be any language that can be recognized in non-deterministic quasilinear time by a random-access machine. Let $x$ be an instance in $L$ of size $N$. Then one can prove that $x$ is in $L$, with perfect completeness and constant soundness, using $\widetilde{\Theta}(\sqrt{N})$ unentangled quantum proofs, each with $\Theta(\log N)$ qubits.
\end{proposition}

More generally, letting $\NTIME_{\RAM}(t)$ be the class of languages solvable in non-deterministic $t(n)$-time by a random-access machine, for any $L$ in $\NTIME_{\RAM}(t)$ it is possible to prove membership in $L$, with perfect completeness and constant soundness, using $\widetilde{\Theta}\big(\sqrt{t(n)}\big)$ unentangled quantum proofs, each with $\Theta\big(\log t(n)\big)$ qubits.

In order to obtain statements analogous to \propositionref{prop:extend-aaronson} for the parameters obtained by other multi-prover $\QMA$ protocols (including \cite{CD10,Bei10,BT09,GallNN12}), we state the ``size-efficient'' reduction from $\NTIME_{\RAM}(t)$ in a very generic form. For details, see \sectionref{sec:quasilinear-PCPs}.

\subsubsection{Improvements to \cite{CD10}}

Aaronson et al. \cite{ABDFS09} raised the question of whether it is possible to construct a (multi-prover) $\QMA$ protocol with constant soundness gap and sublinear proof size for an $\NP$-complete language, but using \emph{no entangled measurements}. Chen and Drucker \cite{CD10} gave a positive answer:

\begin{utheorem}[{\cite{CD10}}]
Let $\varphi$ be a satisfiable $\TSAT$ instance with $N$ variables and $M$ clauses (and $M \geq N$). Then one can prove satisfiability of $\varphi$, with almost-perfect completeness and constant soundness, using $\widetilde{\Theta}(\sqrt{M})$ unentangled quantum proofs, each with $\Theta(\log M)$ qubits, and by only making LOCC (in fact, Bell) measurements.
\end{utheorem}

The analysis of \cite{CD10} does not give a smooth tradeoff between the number of provers and soundness, because their proof only shows a soundness gap when the number of provers is $\widetilde{\Theta}(\sqrt{M})$. We give a tight analysis of their protocol that yields a soundness gap for a number of provers $\NUMP \in \Omega(\log N)$. We believe the smooth trade-off is of interest because it helps us ``push the barrier closer'' to the best two-prover LOCC $\QMA$ protocols with logarithmic proof length and small soundness gap. 

\begin{proposition}
\label{proposition:improved:CD10}
Let $\varphi$ be a satisfiable $\TSAT$ instance with $N$ variables and $M$ clauses (and $M \geq N$). Then one can prove satisfiability of $\varphi$, with completeness $1-e^{-\Omega(\NUMP)}$ and soundness $1-\Omega(\frac{\NUMP^{2}}{N+\NUMP^{2}})$, using $\NUMP$ unentangled quantum proofs, each with $\Theta(\log N)$ qubits, and by only making LOCC (in fact, Bell) measurements. Thus, for $\NUMP \in \Omega(\log N)$ and $\NUMP \in O(\sqrt{N})$, the soundness gap is $\Omega(\NUMP^{2}N^{-1})$. Moreover, our analysis is tight, for reasons described in \remarkref{remark:CD10-soundness-tight} on page \pageref{remark:CD10-soundness-tight}.
\end{proposition}

The proof follows by improving the second-moment argument of \cite{CD10} by using a one-sided Chebyshev inequality. See \sectionref{sec:CD10-improved-soundness} for more details.

\subsubsection{Improvements to \cite{BT09}}

We give a tight soundness analysis of the protocol of Blier and Tapp \cite{BT09}.

\begin{proposition}
\label{proposition:improved:BT09}
The protocol of Blier and Tapp \cite{BT09} for $\TCSP$'s on $N$ vertices and $M$ edge constraints over a $K$-size alphabet has soundness $s = 1 - \Omega(N^{-2})$, assuming $K \in O(1)$. Moreover, our analysis is tight, for reasons described in \remarkref{remark:BT09-soundness-tight} on page \pageref{remark:BT09-soundness-tight}.
\end{proposition}

The above results improves on the original analysis by Blier and Tapp \cite{BT09}, who show that the protocol for instances of graph $3$-coloring on $N$ vertices and $M$ edges that has completeness $c=1$ and soundness $s = 1 - \Omega(N^{-6})$.  It also improves on the result of Beigi~\cite{Bei10}, who gives a protocol for constant-gap instances of (balanced) $2$-out-of-$4$ $\SAT$ with $M$ clauses that has completeness $c = \frac{3}{4} + \frac{\sqrt{2}}{6M}\sqrt{1 - \frac{1}{M}}$ and soundness $s = c - \Omega(M^{-3-\veps})$ for every $\veps > 0$.  In independent work, Le Gall, Nakagawa and Nishimura~\cite{GallNN12} also gave an improved version of the Blier and Tapp \cite{BT09} protocol, by changing the protocol to utilize $\TCSP$ instances with constant soundness gaps, and achieve $c=1$ and $s=1-\Omega(\frac{1}{N})$, which when applied with the requisite PCP results, improves upon our soundness gap by nearly a quadratic factor. See \sectionref{sec:BT09-improved-soundness} for details.

\subsubsection{Conclusions}

Our results have made limited progress in answering \questionref{question:main}, and we now comment on avenues for further progress.

One possible approach is to first construct a two-prover LOCC $\QMA$ protocol for $\TSAT$ with $\widetilde{\Omega}(\frac{1}{\sqrt{N}})$ soundness gap and logarithmic proof size, and then suitably amplify the protocol to constant soundness. Note that since LOCC $\QMA$ protocols amplify naturally, there is no need to use a product test \cite{HM10} (which would have regardless created a non-LOCC protocol), nor there is any need to invoke additional assumptions such as the Weak Additivity Conjecture \cite[Theorem 35]{ABDFS09}.)

Through \propositionref{proposition:improved:BT09} we have made some progress in this direction by improving the soundness gap of two-prover protocols for $\TSAT$ with a polylogarithmic proof size to $\Omega(N^{-2})$. Unfortunately, the protocol is not LOCC and does not achieve the required soundness gap of $\widetilde{\Omega}(\frac{1}{\sqrt{N}})$. 

Another approach is to construct an LOCC protocol that acts on all the provided qubits at the same time, possibly in a much more complicated way than amplifying a two-prover protocol. The main difficulty in such an approach is that one of the main tools (as used in \cite{ABDFS09,KMY03,BT09,Bei10,HM10,GallNN12}) for multi-prover $\QMA$ protocols is the swap test, which is not LOCC. Attempting to replicate the properties of the product test of of Harrow and Montanaro \cite{HM10} (which relies on the swap test) within the LOCC framework (in order to apply it to LOCC protocols such as \cite{CD10}) runs into the risk of implying that $\QMA(\NUMP)=\QMA_{\LOCC}(2)$, which, through the result of Brand\~{a}o et al.\ \cite{BCY10}, would have the unlikely consequence of $\QMA(\NUMP)=\QMA$. Thus, any such approach must make essential ``non-black-box'' use of the structure of the language at hand (e.g., that of $2$-out-of-$4$ $\SAT$) to avoid being a ``generic'' test.

\section{Preliminaries}
\label{sec:preliminaries}

\parhead{Two languages and non-deterministic time} First, we define two $\NP$ languages that we will be working with. The first language is constraint-satisfaction problems on graphs:

\begin{definition}[Graph Constraint Satisfaction]
Let $G,=(V,E)$ be a graph (possibly with self-loops) and an alphabet $\Sigma$. A \emph{graph constraint-satisfaction problem} is a pair $\CTRS=(G,\{R_{e}\}_{e \in E})$ where $R_{e} \colon \Sigma\times\Sigma\to\{0,1\}$ for each $e \in E$. We say that $\CTRS$ is satisfiable if there is a labeling $C \colon V \to \Sigma$ such that every edge predicate evaluates to 1. We say that $\CTRS$ is $\delta$-satisfiable if, for every possibly labeling of the vertices, at most a $\delta$ fraction of the edge predicates evaluate to 1.

Fix positive integers $N$, $M$, and $K$. The class $\TCSP(N,M,K)$ consists of satisfiable graph constraint-satisfaction problems over $K$-size alphabets on graphs of $N$ vertices and $M$ edges.
\end{definition}

The second language is $\SAT$ formulae with some additional structure:

\begin{definition}[$2$-Out-Of-$4$ $\SAT$]
The class $\TOFSAT$ consists of $2$-out-of-$4$ satisfiable $4$-CNF formulae in which every variable appears in $\Theta(1)$ number of clauses. (A $4$-CNF formula is $2$-out-of-$4$ satisfiable if there is an assignment to the variables such that for every clause in the $4$-CNF exactly two of the four variables are satisfied.)

\end{definition}

Next, we recall the definition of a proper complexity function:

\begin{definition}[Proper Complexity Function]
A monotonically increasing function $f \colon \mathbb{N}\to\mathbb{N}$ is a \emph{proper complexity function} if a multi-tape Turing machine can compute $1^n\mapsto 1^{f(n)}$ in time and space $O(f(n))$. See \cite[Definition 7.1]{Pap94} for more details.
\end{definition}

Finally, we recall non-deterministic time complexity classes with respect to multi-tape Turing machines and random-access machines:

\begin{definition}[$\NTIME_{\MTM}$]
A multi-tape Turing machine is a finite state machine attached to multiple tapes, with one head per tape. The tapes are infinite in one direction. The machine can read and write to each tape, moving one cell per time step. See \cite{Pap94} for more details.

That the machine is non-deterministic means that the finite state control of the Turing machine can non-deterministically decide its next move, such that the machine accepts if and only if there is some non-deterministic choice that allows it to accept.

For a proper complexity function $t$, $\NTIME_{\MTM}(t)$ denotes the class of languages that can be recognized by a $t$-time non-deterministic multi-tape Turing machine.
\end{definition}

\begin{definition}[$\NTIME_{\RAM}$]
A random-access machine (RAM) is a list of commands that includes a finite number of control registers as well as an unbounded number of indexable registers.  Each register holds an integer. Commands include addition, multiplication (with a log-cost penalty), branching on register contents, and indexing the registers with the contents of other registers. See \cite{GS89} for more details.

That the machine is non-deterministic means that the finite list of commands of the random-access machine allow non-deterministic branching to its next move, such that the machine accepts if and only if there is some non-deterministic choice that allows it to accept.

For a proper complexity function $t$, $\NTIME_{\RAM}(t)$ denotes the class of languages that can be recognized by a $t$-time non-deterministic random-access machine.
\end{definition}

\parhead{Information theory}
First, we recall the classical information-theoretic notion of statistical distance between two probability distributions \cite[Sec. 9.1]{NC00}:

\begin{definition}[Statistical Distance]
\label{def:statistical-distance}
Let $P$ and $Q$ be two probability distributions over the same finite set $S$. The \emph{statistical distance} between $P$ and $Q$, denoted $\VARNORM{P-Q}$, is defined as the quantity
\begin{equation*}
\VARNORM{P-Q} := \dfrac{1}{2} \sum_{s \in S} \big\vert P(s) - Q(s) \big\vert \enspace.
\end{equation*}
\end{definition}

Next, we recall its quantum analogue of trace distance \cite[Sec. 9.2.1]{NC00}:

\begin{definition}[Trace Distance]
\label{def:trace-distance}
The \emph{trace distance} between two quantum states $\rho$ and $\sigma$, denoted $\TRACENORM{\rho-\sigma}$, is defined as the quantity:
\begin{equation*}
\TRACENORM{\rho-\sigma}:=\dfrac{1}{2} \tr \left( \sqrt{(\rho - \sigma)^{\dagger}(\rho - \sigma)} \right) \enspace.
\end{equation*}
\end{definition}

If $\rho$ and $\sigma$ commute, then they can be simultaneously diagonalized, and the trace distance between $\rho$ and $\sigma$ reduces to the statistical distance between the two probability distributions induced by the two sets of eigenvalues of $\rho$ and $\sigma$.

If $\rho$ and $\sigma$ are two pure states $\ketbra{\phi}{\phi}$ and $\ketbra{\psi}{\psi}$, then the trace distance between $\ketbra{\phi}{\phi}$ and $\ketbra{\psi}{\psi}$ simplifies \cite[Sec. 9.2.3]{NC00} to the quantity $\sqrt{1-\abs{\braket{\phi}{\psi}}^{2}}$. 

Also, we recall that, given a projective measurement (i.e., a Hermitian operator) $M$, if $P$ and $Q$ are the probability distributions describing the outcomes obtained when measuring $M$ on $\ket{\phi}$ and $\ket{\psi}$ respectively, then $\VARNORM{P-Q} \leq \TRACENORM{\ketbra{\phi}{\phi}-\ketbra{\psi}{\psi}}$. For convenience, we will denote by $\DIST_{M}(\ket{\phi})$ the distribution $P$, and simply $\DIST(\ket{\phi})$ when $M$ is assumed to be a full computational basis measurement.

\parhead{Swap test}
The  \emph{swap-test on two quantum states $\rho$ and $\sigma$} \cite{BarencoBDEJM97,BCWdW01} is given by the following quantum circuit:
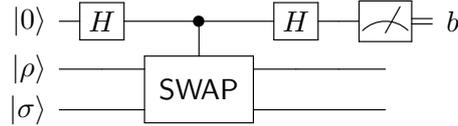
\begin{equation*}
\Qcircuit @C=.7em @R=.5em
{
\lstick{\ket{0}}      & \gate{H} & \ctrl{1}                     & \gate{H} & \qw & \meter & \cw & b  \\
\lstick{\ket{\rho}}   & \qw      & \multigate{1}{\mathsf{SWAP}} & \qw      & \qw & \qw    &     &    \\
\lstick{\ket{\sigma}} & \qw      & \ghost{\mathsf{SWAP}}        & \qw      & \qw & \qw    &     &
}
\end{equation*}
Essentially, the swap-test measures the overlap between two quantum states, because
\begin{equation*}
\Pr \big[ b=0 \big] = \dfrac{1 + \tr (\rho \sigma)}{2} \enspace,
\end{equation*}

If $\rho$ and $\sigma$ are two pure states $\ketbra{\phi}{\phi}$ and $\ketbra{\psi}{\psi}$, then the probability above is equal to $\frac{1 + \abs{\braket{\phi}{\psi}}^{2}}{2}$. Interpreting $b=0$ as an ``accept'' and $b=1$ as a ``reject'', we define:
\begin{equation*}
\REJ\big(\SWAPTEST(\ket{\phi},\ket{\psi})\big) := \frac{1 - \abs{\braket{\phi}{\psi}}^{2}}{2} \enspace.
\end{equation*}
The swap test can thus be used to check whether $\ketbra{\phi}{\phi}$ and $\ketbra{\psi}{\psi}$ are equal or not: if they are equal, then $\REJ\big(\SWAPTEST(\ket{\phi},\ket{\psi})\big)=0$; if they are not equal, then the probability of rejection is inversely proportional to the overlap between $\ketbra{\phi}{\phi}$ and $\ketbra{\psi}{\psi}$.

If the probability that the swap test rejects two quantum states is bounded from above, then the statistical distance between the two probability distributions arising when measuring the two quantum states (in any basis) is also bounded from above:

\begin{lemma}
\label{lemma:swap-test-and-measurements}
Let $\ket{\phi}$ and $\ket{\psi}$ be quantum states and $\delta$ a number in $[0,1]$. If $\REJ\big(\SWAPTEST(\ket{\phi},\ket{\psi})\big) \leq \delta$, then $\VARNORM{\DIST(\ket{\phi})-\DIST(\ket{\psi})} \leq \sqrt{2\delta}$.
\end{lemma}

\begin{proof}
Recall that $\REJ\big(\SWAPTEST(\ket{\phi},\ket{\psi})\big) = \frac{1}{2} - \frac{|\braket{\phi}{\psi}|^{2}}{2}$, so that
\begin{equation*}
     \VARNORM{\DIST(\ket{\phi})-\DIST(\ket{\psi})}
\leq \TRACENORM{\ketbra{\phi}{\phi} - \ketbra{\psi}{\psi}}
=    \sqrt{1 - |\braket{\phi}{\psi}|^{2}}
\leq \sqrt{2\delta}
\enspace.
\end{equation*}
\end{proof}

\parhead{Quantum Fourier transform}
Finally, we recall the quantum Fourier transform:

\begin{definition}[Quantum Fourier Transform]
Let $\mathcal{H}_{n}$ be an $n$-dimensional Hilbert space with orthonormal basis $\{\ket{0},\ldots,\ket{n-1}\}$. The \emph{$n$-dimensional quantum Fourier transform}, denoted $F_{n}$, is the linear operator whose action on the basis vector $\ket{j}$, $j \in \{0,1\ldots,n-1\}$, is given by
\begin{equation*}
\ket{j} \mapsto \dfrac{1}{\sqrt{n}} \cdot \sum_{k=0}^{n-1} e^{2\pi \sqrt{-1} (jkn^{-1})} \ket{k} \enspace.
\end{equation*}
\end{definition}

We recall that $F_{n}$ is a unitary operator. Also, we will denote by $\ket{0_{F_{n}}}$ the image of $\ket{0}$ under $F_{n}$, and it satisfies the following equation:
\begin{equation*}
\ket{0_{F_{n}}} = \dfrac{\ket{0} + \ket{1} + \cdots + \ket{n-1}}{\sqrt{n}} \enspace.
\end{equation*}
For more details on the quantum Fourier transform, see \cite[Ch. 5]{NC00}.

\parhead{Quantum proofs}
In an $\NUMP$-prover $\QMA$ protocol, the $\BQP$ verifier (Arthur) receives a classical input $x \in \{0,1\}^{*}$ together with $\NUMP$ quantum proofs $\ket{\BPSI^{(1)}},\ldots,\ket{\BPSI^{(\NUMP)}}$ sent by the $\NUMP$ provers (Merlins); the provers (and thus the quantum proofs they send) are promised to be \emph{unentangled}. The verifier will then decide whether to accept $x$ or not, based on all the quantum proofs he received.

\begin{definition}[Multi-Prover $\QMA$]
Fix a set of projectors $\mathcal{M}$, polynomially-bounded functions $\NUMP,\PRL \colon \mathbb{N} \to \mathbb{N}$, and arbitrary functions $s,c \colon \mathbb{N} \to [0,1]$.

A language $L \subseteq \{0,1\}^{*}$ is in $\QMA_{\PRL}^{\mathcal{M}}(\NUMP,c,s)$ if there exists a polynomial-time quantum algorithm $V_{L}$ restricted to performing measurements from $\mathcal{M}$ such that, for all inputs $x \in \{0,1\}^{n}$, the following conditions hold:
\begin{itemize}
  \item \emph{Completeness:} If $x \in L$, there exist $\NUMP(n)$ quantum proofs $\ket{\BPSI^{(1)}},\ldots,\ket{\BPSI^{(\NUMP(n))}}$, each a state of at most $\PRL(n)$ qubits, such that $V_{L}$ accepts with probability at least $c(n)$ on input $\ket{x}\ket{\BPSI^{(1)}} \cdots \ket{\BPSI^{(\NUMP(n))}}$; and
  \item \emph{Soundness:} If $x \not\in L$, for every $\NUMP(n)$ quantum proofs $\ket{\BPSI^{(1)}},\ldots,\ket{\BPSI^{(\NUMP(n))}}$, each a state of at most $\PRL(n)$ qubits, $V_{L}$ accepts with probability at most $s(n)$ on input $\ket{x}\ket{\BPSI^{(1)}} \cdots \ket{\BPSI^{(\NUMP(n))}}$.
\end{itemize}
The class $\QMA_{\PRL}(\NUMP,c,s)$ is defined to be the class $\QMA_{\PRL}^{\mathcal{M}}(\NUMP,c,s)$ where $\mathcal{M}$ is the set of all Hermitian operators. The class $\QMA(\NUMP,c,s)$ is defined to be the class $\QMA_{\poly(\cdot)}(\NUMP,c,s)$. The class $\QMA(\NUMP)$ is defined to be the class $\QMA(\NUMP,2/3,2/3)$.
\end{definition}

Note that any set of admissible Hermitian operators $\mathcal{M}$ induces a set of binary measurements, where each $M \in \mathcal{M}$ means ``accept'' and $I - M$ means ``reject''. For example, $\mathcal{M} = \BELL$ is the set of Bell measurements (non-adaptive, unentangled measurements),  $\mathcal{M} = \LOCC$ is the set of LOCC measurements (adaptive, unentangled measurements), and $\mathcal{M} = \SEP$ is the set of separable measurements (which includes the swap test and product test).

\section{Quasilinear-Time Has Sublinear Unentangled Quantum Proofs}
\label{sec:quasilinear-PCPs}

In this section, we show how a few simple observations suffice to generalize the known positive results on multi-prover $\QMA$ protocols for $\NP$ languages (i.e., \cite{BT09}, \cite{Bei10}, \cite{ABDFS09}, \cite{CD10}, and \cite{HM10}). Doing so allows us to exhibit a large class of problems that qualify as positive examples to \questionref{question:ABDFS09} raised by Aaronson et al. \cite{ABDFS09}.

The main observation is the fact that ``short'' PCPs exist not only for $\TSAT$ but, more generally, for every $\NTIME$ language:

\begin{claim}[Quasilinear PCPs for $\NTIME$ Languages]
\label{claim:quasilinear-PCPs}
Let $t \colon \mathbb{N} \to \mathbb{N}$ be any proper complexity function and let $L$ be a language in $\NTIME_{\RAM}(t)$. Then, there exist
\begin{itemize}[nolistsep]
  \item a \emph{size function} $S_{L} \colon \mathbb{N} \to \mathbb{N}$ with $S_{L}(n) \in \widetilde{O}(t(n))$,
  \item a \emph{density function} $D_{L} \colon \mathbb{N} \to \mathbb{N}$ with $D_{L}(n) \in\widetilde{O}(t(n))$,
  \item a \emph{color constant} $K_{L} \in \mathbb{N}$,
  \item a \emph{reduction complexity function} $C_{L} \colon \mathbb{N} \to \mathbb{N}$ with $C_{L} \in \poly(t(n))$,
  \item a $C_{L}$-time reduction $R_{L} \colon \{0,1\}^{*} \to \{0,1\}^{*}$ from $L$ to $\TCSP$,
  \item a \emph{gap constant} $\eta_{L} \in (0,1)$, and
  \item a \emph{regularity constant} $d_{L} \in \mathbb{N}$,
\end{itemize}
such that, for every instance $x \in \{0,1\}^{n}$, the following properties hold:
\begin{itemize}[nolistsep]
  \item \emph{Efficiency:} $R_{L}(x)$ is a $\TCSP(N,M,K)$ instance with $N = S_{L}(n)$, $M = D_{L}(n)$, and $K = K_{L}$; 
  \item \emph{Completeness:} if $x \in L$, then $R_{L}(x) \in \TCSP(N,M,K)$;
  \item \emph{Soundness:} if $x \not\in L$, then $R_{L}(x)$ is $(1-\eta_{L})$-satisfiable; and
  \item \emph{Regularity:} $R_{L}(x)$ is $d_{L}$-regular (with self-loops).
\end{itemize}
\end{claim}

The claim is simply a statement about the short PCPs that are obtained from the works of Ben-Sasson and Sudan \cite{BS05} and Dinur \cite{Din07}, together with some simple observations about the generality of their results.

\begin{proof}
Ben-Sasson and Sudan proved \cite[Theorem 2.2]{BS05} that
\begin{equation*}
\NTIME_{\MTM}(t(n)) \subseteq \PCP_{1,\frac{1}{2}} \left[ \log\big(t(n) \polylog(t(n))\big) \,,\, \polylog\big(t(n)\big) \right] \enspace.
\end{equation*}
In order to construct ``short'' PCPs for $\TSAT$, Dinur:
\begin{itemize}
  \item[(i)] observes \cite[Lemma 8.3]{Din07} that the result of Ben-Sasson and Sudan implies that $\TSAT$ can be reduced to $\TCSP$ instances that satisfy all the properties of the claim, with the exception that the soundness gap is is only $1/\polylog(n)$; and
  \item[(ii)] then she applies her main technical result of gap amplification \cite[Theorem 1.5]{Din07} to bring the soundness gap to a constant $\eta$. 
\end{itemize}
Then, (i) and (ii) together easily imply ``short'' PCPs for $\TSAT$ \cite[Theorem 8.1]{Din07}.

We note that Dinur's first observation, $(i)$, only relies on the fact that $\TSAT \in \NTIME_{\MTM}(t(n))$ for some $t(n) \in \widetilde{O}(n)$, and a similar observation can be made for a general language $L \in \NTIME_{\MTM}(t(n))$, which, again combined with her gap amplification result, yields the claim for languages in $\NTIME_{\MTM}(t(n))$.

We choose to not state the claim for languages in $\NTIME_{\MTM}(t(n))$, because it is not as illuminating; it seems quite tedious (and difficult) to check whether a given language $L$ can be recognized in non-deterministic $t(n)$-time by some multi-tape Turing machine. Instead, we observe that, by using a result of Gurevich and Shelah \cite[Theorem 2]{GS89}, which implies that $\NTIME_{\RAM}(t) \subseteq \NTIME_{\MTM}( t'(n) )$ for some $t'(n) \in \widetilde{O}(t(n))$, we obtain the claim as stated;\footnote{The result of Gurevich and Shelah \cite[Theorem 2]{GS89} in fact states that $\NTIME_{\RAM}(n) \subseteq \NTIME_{\MTM}( \widetilde{O}(n) )$, but the proof can easily be extended to show that $\NTIME_{\RAM}(t) \subseteq \NTIME_{\MTM}( t'(n) )$ for some $t'(n) \in\widetilde{O}(t(n))$.} this way, the task of checking whether a given language is in $L$ is in $\NTIME_{\RAM}(t)$ is much simpler: one only needs to write ``pseudocode'' for the non-deterministic verifier, and prove that it halts in time $t(n)$.
\end{proof}

The $\TCSP$ instance guaranteed by \claimref{claim:quasilinear-PCPs} is already a ``nice'' instance for which multiple-prover $\QMA$ results have been proved. For example, a $\TCSP$ instance is the starting point of Blier and Tapp \cite{BT09} and Chen and Drucker \cite{CD10}, so that we obtain generic results for both protocols.\footnote{Blier and Tapp, though, do not exploit the constant soundness gap, so they simply start from the classical $\NP$-complete problem of graph $3$-colorability.} 

Other works, instead, such as \cite{ABDFS09} and \cite{Bei10} ``process the $\TCSP$ instance further'', in order to give it additional structure (that is exploited in their protocols). Thus, these additional processings also inherit the more general reduction guaranteed by \claimref{claim:quasilinear-PCPs} for all of the languages in $\NTIME_{\RAM}(t)$:

\begin{corollary}[Constant-Gap Boolean Formulae for $\NTIME$ Languages]
\label{for:constant-gap}
Let $t \colon \mathbb{N} \to \mathbb{N}$ be any proper complexity function and let $L$ be a language in $\NTIME_{\RAM}(t)$. Then,
\begin{itemize}

  \item[(i)] \emph{\textbf{Constant-Gap $\TSAT$ Formulae:}} there exist
  \begin{itemize}[nolistsep,leftmargin=13pt]
    \item a \emph{size function} $S_{L} \colon \mathbb{N} \to \mathbb{N}$ with $S_{L}(n) \in \widetilde{O}(t(n))$,
    \item a \emph{density function} $D_{L} \colon \mathbb{N} \to \mathbb{N}$ with $D_{L}(n) \in \widetilde{O}(t(n))$,
    \item a \emph{reduction complexity function} $C_{L} \colon \mathbb{N} \to \mathbb{N}$ with $C_{L} \in \poly(t(n))$,
    \item a $C_{L}$-time reduction $R_{L} \colon \{0,1\}^{*} \to \{0,1\}^{*}$ from $L$ to $\TSAT$, and
    \item a \emph{gap constant} $\eta_{L} \in (0,1)$,
  \end{itemize}
  such that, for every instance $x \in \{0,1\}^{n}$, the following properties hold:
  \begin{itemize}[nolistsep,leftmargin=13pt]
    \item \emph{Efficiency:} $R_{L}(x)$ is a $\TSAT$ instance with $N = S_{L}(n)$ variables and $M = D_{L}(n)$ clauses;
    \item \emph{Completeness:} if $x \in L$, then $R_{L}(x) \in \TSAT$; and
    \item \emph{Soundness:} if $x \not\in L$, then $R_{L}(x)$ is $(1-\eta_{L})$-satisfiable.
  \end{itemize}
  
  \item[(ii)] \emph{\textbf{Constant-Gap $\TOFSAT$ Formulae:}} there exist
  \begin{itemize}[nolistsep,leftmargin=13pt]
    \item a \emph{size function} $S_{L} \colon \mathbb{N} \to \mathbb{N}$ with $S_{L}(n) \in \widetilde{O}(t(n))$,
    \item a \emph{density function} $D_{L} \colon \mathbb{N} \to \mathbb{N}$ with $D_{L}(n) \in \widetilde{O}(t(n))$,
    \item a \emph{reduction complexity function} $C_{L} \colon \mathbb{N} \to \mathbb{N}$ with $C_{L} \in \poly(t(n))$,
    \item a $C_{L}$-time reduction $R_{L} \colon \{0,1\}^{*} \to \{0,1\}^{*}$ from $L$ to $\TOFSAT$, and
    \item a \emph{gap constant} $\eta_{L} \in (0,1)$,
    \item a \emph{balance constant} $b_{L} \in \mathbb{N}$,
  \end{itemize}
  such that, for every instance $x \in \{0,1\}^{n}$, the following properties hold:
  \begin{itemize}[nolistsep,leftmargin=13pt]
    \item \emph{Efficiency:} $R_{L}(x)$ is a $\TOFSAT$ instance with $N = S_{L}(n)$ variables and $M = D_{L}(n)$ clauses;
    \item \emph{Completeness:} if $x \in L$, then $R_{L}(x) \in \TOFSAT$;
    \item \emph{Soundness:} if $x \not\in L$, then $R_{L}(x)$ is $(1-\eta_{L})$-satisfiable;
    \item \emph{Balance:} each variable of $R_{L}(x)$ appears in at most $b_{L}$ clauses.
  \end{itemize}
  
\end{itemize}
\end{corollary}

Of course, one could add more items to the above corollary, other than $\TSAT$ and $\TOFSAT$, if other languages that can be efficiently reduced to from $\TCSP$ are found to be useful. We chose to mention only $\TSAT$ because of its general importance and $\TOFSAT$ because it was successfully used by \cite{ABDFS09} and \cite{Bei10}.

The proof of the corollary was partly sketched, in the particular case of $t(n)=n$ in \cite[Lemma 12]{ABDFS09}. We give here the more general proof:

\begin{proof}[Proof of \corollaryref{for:constant-gap}]
To obtain (i), we argue as follows. To prove the first item, it suffices to convert the instance guaranteed by \claimref{claim:quasilinear-PCPs}, which is a $\TCSP$ instance over a constant-size alphabet, into a $\TSAT$ instance, in a way that preserves perfect completeness and degrades the soundness gap by at most a constant factor.

First, consider a $\TCSP$ instance over a constant-size alphabet. Observe that we can transform this into a $\CSP$ over a binary alphabet by allowing constraints to restrict multiple variables. As the original alphabet was of constant-size, this only increase the arity, number of variables, and number of constraints in the $\CSP$ by a constant factor. Further, the soundness gap is preserved.  

So consider now a constraint $C$ in the $\CSP$ over variables $\vec{x}$. By the Cook-Levin Theorem, there exist a $\TSAT$ formula $\varphi_{C}$ and additional variables $\vec{y}$ such that $C(\vec{x})$ if and only if there exists $\vec{y}$ such that $\varphi_C(\vec{x},\vec{y})$. Observe that the size of $\varphi_{C}$ is at most some constant $g$, because the original $\CSP$ is over a constant-size alphabet and has arity $2$. Define the output of this reduction to be the $\TSAT$ formula $\varphi:=\bigwedge_C \varphi_C$. 

We now analyze the properties of $\varphi$. First, observe that the number of clauses in $\varphi$ is at most $g$ times the number of constraints in the original $\CSP$, and the number of variables is also a constant-factor more than the number of variables in the original $\CSP$.  Further, if the original $\CSP$ was satisfiable, then so must be $\varphi_{C}$, so perfect completeness is preserved. To analyze soundness, suppose that the original $\CSP$ was at most $\delta$ satisfiable. Then, in any assignment to $\varphi$, at least $(1-\delta)\cdot E$ clauses must be unsatisfied, where $E$ is the number of constraints in the original $\CSP$. As there are at most $gE$ clauses in $\varphi$, this means that $\varphi$ can have at most a $(1-\frac{1-\delta}{g})$-fraction of satisfied clauses. Thus, there is still a constant soundness gap.

To obtain (ii), we first invoke (i) so to obtain a reduction to $\TSAT$, and then follow the outline of Aaronson et al. \cite[Lemma 12]{ABDFS09}. Specifically, the instance output by the reduction guaranteed by (i) can first be further modified using a reduction of Papadimitriou and Yannakakis \cite{PY91} from $\TSAT$ to $\TSAT$ that makes each variable appear in at most $b_{L}=29$ (in fact, exactly) clauses (and this reduction preserves the constant soundness gap); then, we apply a reduction of Khanna et al. \cite{KSTW01} (that preserves both the constant soundness gap and the balanced property of the formula) from $\TSAT$ to $\TOFSAT$. The reason that the outline of Aaronson et al. \cite[Lemma 12]{ABDFS09} also works in the general case considered in this corollary is that the number of variables and clauses increases only a by a constant through these two additional reductions.
\end{proof}

By combining the above results with \cite{ABDFS09}, we have now established \propositionref{proposition:sublinear}.

\section{Graph Coloring States}
\label{sec:graph-states}

Let $G=(V,E)$ be a graph with $N$ vertices and $M$ edges, and let $\Sigma$ be a color alphabet of size $K$. The graph $G$ and the color alphabet $\Sigma$ will be fixed throughout the rest of the paper.

We say that a quantum state $\ket{\BPSI}$ is a \emph{graph coloring state} (for $G$ and $\Sigma$) if it is a quantum state over a Hilbert space $\HILBERT = (\HILBERT_{2})^{\otimes \log N} \otimes (\HILBERT_{2})^{\otimes \log K}$. Thus, any graph coloring state $\ket{\BPSI}$ can be written as
\begin{equation*}
\ket{\BPSI} = \sum_{v=0}^{N-1} \alpha_{v} \ket{v} \sum_{j=0}^{K-1} \beta_{v,j} \ket{j} \enspace,
\end{equation*}
where $\sum_{v=0}^{N-1} \abs{\alpha_{v}}^{2}=1$ and $\sum_{j=0}^{K-1} \abs{\beta_{v,j}}^{2}=1$ for each $v \in \{0,\ldots,N-1\}$. Note that the definition of a graph coloring state is independent of the edge set $E$.  Such a state is intended to allow the provers to honestly encode a coloring $\chi:V\to\Sigma$ of the graph via the state \[\ket{\BPSI}=\frac{1}{\sqrt{N}}\sum_{v=0}^{N-1}\ket{v}\ket{\chi(v)}\;.\]  The challenge of using these states is to enforce the provers to act as in the honest case.  If they encode the coloring as above, we can recover/check it by measuring the state and recovering vertex/color pairs. If the coloring is invalid then with measurements of independent (identical) states we can observe invalidly colored edges.  However, dishonest provers could allow many colors for each vertex, or make it so some vertices have zero amplitudes.  Doing such things would fool the above coloring test.  Thus, the main challenge is to detect this dishonest case and reject it with good probability.

We believe that developing strong tools for \emph{quantum property testing} is essential for making improvement towards better multi-prover $\QMA$ protocols.\footnote{For example, we believe that a two-prover $\QMA_{\LOCC}$ protocol for $\TCSP(N,\widetilde{O}(N),O(1))$ with $\Omega(1/N)$ soundness gap and polylogarithmic proof length exists, but we do not know of one yet. Our improved soundness analysis of \cite{CD10} almost achieves that, and it seems that a somewhat smarter LOCC verifier should suffice. Developing a theory of quantum property testing should shed some light on how to design such a verifier.} As a first move in that direction, we give in the next subsection two lemmas for graph coloring states, which were implicitly used in both \cite{BT09} and \cite{CD10} with very different parameters, and present them in a generic form. After that, we summarize the tests for graph coloring states that have been used in previous protocols.

\subsection{Two Lemmas on Graph Coloring States}

Let us first introduce some simple notation: given any graph coloring state $\ket{\BPSI}$,
\begin{itemize}
  \item for $c \in (0,1]$, $R_{c}(\ket{\BPSI})$ is the subset of $V$ consisting of those vertices $v$ for which $\abs{\alpha_{v}}^{2} < c$;
  \item for $c \in (0,1]$, $S_{c}(\ket{\BPSI})$ is equal to $V - R_{c}(\ket{\BPSI})$;
  \item for $j=0,\dots,K-1$, $p_{j}(\ket{\BPSI})$ is equal to the probability of measuring $j$ in the color register of the quantum state $(I_{N} \otimes F_{K}) \ket{\BPSI}$; and
  \item for $j=0,\dots,K-1$, $\ket{\gamma(j)} = \sum_{v=0}^{N-1} \gamma_{v}(j) \ket{v}$ is the reduced quantum state obtained when we measure $j$ in the color register of $(I_{N} \otimes F_{K})\ket{\BPSI}$.
\end{itemize}

First, we prove that, as long as a color $j$ has a large-enough probability of being measured in the color register of $(I_{N} \otimes F_{K}) \ket{\BPSI}$, if a vertex $v$ has small amplitude then it will also have a small amplitude in the reduced state conditioned on measuring $j$.

\begin{lemma}[modified {\cite[Lemma 3]{CD10}}, which was implicit in {\cite[Lemma 3.7]{BT09}}]
\label{lemma:gen-uniformity}
Fix a vertex $v \in \{0,\ldots,N-1\}$, a color $j \in \{0,\ldots,K-1\}$, and two positive numbers $c_{1}$ and $c_{2}$. Then:
\begin{equation*}
\left(
p_{j}\big(\ket{\BPSI}\big) \geq \frac{1}{c_{1}} \text{ and }
\abs{\alpha_{v}}^{2} < \frac{1}{c_{2}N}
\right)
\longrightarrow
\left(
\Abs{\gamma_{v}(j)}^{2} < \frac{c_{1}}{c_{2}N}
\right)
\enspace.
\end{equation*}
\end{lemma}

\begin{proof}
Let $\ket{\BCHI}$ be the quantum state obtained from $\ket{\BPSI}$ after performing the quantum Fourier transform on the color register of $\ket{\BPSI}$, i.e., 
\begin{align*}
     \ket{\BCHI}
=\,& (I_{N} \otimes F_{K}) \ket{\BPSI} \\
=\,& (I_{N} \otimes F_{K}) \sum_{v=0}^{N-1} \alpha_{v} \ket{v} \sum_{j=0}^{K-1} \beta_{v,j} \ket{j} \\
=\,& \sum_{v=0}^{N-1} \alpha_{v} \ket{v} \sum_{j=0}^{K-1} \beta_{v,j} \dfrac{1}{\sqrt{K}} \sum_{k=0}^{K-1} e^{\frac{2 \pi \sqrt{-1} jk}{K}} \ket{k} \\
=\,& \dfrac{1}{\sqrt{K}} \sum_{k=0}^{K-1} \left( \sum_{v=0}^{N-1} \alpha_{v} \left( \sum_{j=0}^{K-1} \beta_{v,j} e^{\frac{2 \pi \sqrt{-1} jk}{K}} \right) \ket{v} \right) \ket{k}
\enspace.
\end{align*}
For each $v \in \{0,\ldots,N-1\}$, let $P_{v,j}(\ket{\BPSI})$ be the probability that the color register of $\ket{\BCHI}$ is measured $j$ and the vertex register is measured $v$. Recalling that $\ket{\gamma(j)} = \sum_{v=0}^{N-1} \gamma_{v}(j) \ket{v}$ is the reduced quantum state when outcome $j$ occurs, we have that
\begin{equation*}
P_{v,j}(\ket{\BPSI}) = p_{j}(\ket{\BPSI}) \cdot \Abs{\gamma_{v}(j)}^{2} \enspace.
\end{equation*}
On the other hand, it is also the case that
\begin{align*}
         P_{v,j}(\ket{\BPSI})
=   \,&  \Abs{ \dfrac{\alpha_{v}}{\sqrt{K}} \sum_{j=0}^{K-1} \beta_{v,j} e^{\frac{2 \pi \sqrt{-1} jk}{K}} }^{2} \\
=   \,& \dfrac{\Abs{\alpha_{v}}^{2}}{K} \cdot \Abs{\sum_{j=0}^{K-1} \beta_{v,j} e^{\frac{2 \pi \sqrt{-1} jk}{K}}}^{2} \\
\leq\,& \dfrac{\Abs{\alpha_{v}}^{2}}{K} \cdot K \sum_{j=0}^{K-1} \Abs{\beta_{v,j} e^{\frac{2 \pi \sqrt{-1} jk}{K}}}^{2} \hspace{1cm}\text{(by Cauchy--Schwarz)} \\
=   \,& \Abs{\alpha_{v}}^{2}
\enspace.
\end{align*}
We deduce that $p_{j}(\ket{\BPSI}) \cdot \Abs{\gamma_{v}(j)}^{2} \leq \Abs{\alpha_{v}}^{2}$ or, equivalently, that
\begin{equation*}
\Abs{\gamma_{v}(j)}^{2} \leq \dfrac{\Abs{\alpha_{v}}^{2}}{p_{j}(\ket{\BPSI})} \enspace.
\end{equation*}
By assumption, the probability of measuring $j$ in the color register of $\ket{\BCHI}=(I_{N} \otimes F_{K})\ket{\BPSI}$, which is $p_{j}\big(\ket{\BPSI}\big)$, is at least $\frac{1}{c_{1}}$. Also by assumption, $\abs{\alpha_{v}}^{2} < \frac{1}{c_{2}N}$. Therefore,
\begin{equation*}
\Abs{\gamma_{v}(j)}^{2} \leq \dfrac{\Abs{\alpha_{v}}^{2}}{p_{j}(\ket{\BPSI})} < \dfrac{c_{1}}{c_{2}N} \enspace,
\end{equation*}
as desired.
\end{proof}

Next, we prove that if a quantum state has at least one amplitude that is ``far'' from uniform, then the probability of measuring any given outcome in the Fourier basis can be upper bounded.

\begin{lemma}
\label{lemma:vertex-and-Fourier}
Let $\ket{\gamma} = \sum_{w=0}^{N-1} \gamma_{w} \ket{w}$ be a quantum state. For every $v \in \{0,\ldots,N-1\}$, the probability of measuring $v$ in the (only) register of $F_{N}\ket{\gamma}$ is at most
\begin{equation*}
1 - \dfrac{1}{4} \left( \sum_{w=0}^{N-1} \Abs{ \Abs{\gamma_{w}}^{2} - \dfrac{1}{N} } \right)^{2} \enspace.
\end{equation*}
\end{lemma}

\begin{proof}
The probability of measuring $v$ in the (only) register of $F_{N}\ket{\gamma}$ is given by
\begin{equation*}
\Abs{\bra{\gamma}F_{N}^\dagger\ket{v}}^{2} \enspace.
\end{equation*}
Observe that
\begin{align*}
       \VARNORM{\DIST\big(\ket{\gamma}\big)-\DIST\big(F_{N}^\dagger\ket{v}\big)}
=   \, \dfrac{1}{2} \sum_{w=0}^{N-1} \Abs{ \Abs{\gamma_{w}}^{2} - \Abs{\dfrac{e^{-\frac{2\pi\sqrt{-1}wv}{N}}}{\sqrt{N}}}^{2} }
=   \, \dfrac{1}{2} \sum_{w=0}^{N-1} \Abs{ \Abs{\gamma_{w}}^{2} - \dfrac{1}{N} }
\enspace.
\end{align*}
Recalling that $\VARNORM{\DIST(\ket{\phi})-\DIST(\ket{\psi})} \leq \sqrt{1 - \abs{\braket{\phi}{\psi}}^{2}}$, we obtain that
\begin{equation*}
\Abs{\bra{\gamma}F_{N}^\dagger\ket{v}}^{2} \leq 1 - \dfrac{1}{4} \left( \sum_{w=0}^{N-1} \Abs{ \Abs{\gamma_{w}}^{2} - \dfrac{1}{N} } \right)^{2} \enspace,
\end{equation*}
as desired.
\end{proof}

\subsection{Summary of Tests for Graph Coloring States}

We give a brief summary and description of the tests that have been used successfully in protocols with graph coloring states. The first one is the \textbf{swap test}, which checks whether two states are close to each other:
\begin{itemize}
  \item[] $\SWAPTEST\big(\ket{\BPSI},\ket{\BPHI}\big) \equiv$
  \begin{enumerate}[nolistsep]
    \item Perform the swap test on the two quantum (graph) states $\ket{\psi}$ and $\ket{\phi}$.
    \item Accept if and only if the swap test accepts.
  \end{enumerate}
\end{itemize}
The swap test performs a superposition of swapping the two states, and not swapping the states.  By then combining these superpositions, the interference will leave a result proportional to the distance of the two original states. See \sectionref{sec:preliminaries} for the details and properties of the swap test. Another test that is often useful is the \textbf{uniformity test}:
\begin{itemize}
  \item[] $\UNITEST\big(\ket{\BPSI}\big) \equiv$
  \begin{enumerate}[nolistsep]
    \item Compute $\ket{\BPHI} = (F_{N} \otimes F_{K}) \ket{\BPSI}$.
    \item Measure the vertex and color register of $\ket{\BPHI}$ in the computational basis to get outcome $(v,j)$.
    \item If $j=0$ but $v \neq 0$, then reject.
    \item Accept.
  \end{enumerate}
\end{itemize}
The uniformity test seeks to ensure that the total amplitude of each vertex is large, assuming that the probability of measuring $j=0$ is large.  This is used in ensuring that a graph coloring state meaningfully assigns a color to each vertex in the graph.  A generalization of this test is the \textbf{conditional uniformity test}: for any $z \in [0,\kappa]$,
\begin{itemize}
  \item[] $\CONDUNITEST_{z}\big(\ket{\BPSI^{(1)}},\ldots,\ket{\BPSI^{(\NUMP)}}\big) \equiv$
  \begin{enumerate}[nolistsep]
    \item For $i=1,\ldots,\NUMP$, compute $\ket{\BPHI^{(i)}} = (F_{N} \otimes F_{K}) \ket{\BPSI^{(i)}}$.
    \item For $i=1,\ldots,\NUMP$, measure the vertex and color register of $\ket{\BPHI^{(i)}}$ in the computational basis to get outcome $(v_{i},j_{i})$.
    \item If $z  > \Abs{\{ i \in \{1,\ldots,\NUMP\} \,:\, j_{i}=0 \}}$, then reject.\label{cut}
    \item For $i=1,\ldots,\NUMP$, if $j_{i}=0$ but $v_{i} \neq 0$, then reject.
    \item Accept.
  \end{enumerate}
\end{itemize}
Intuitively, the conditional uniformity test also makes sure that a significant fraction of the graph coloring states are such that, when their color register is measured in the Fourier basis, the color $0$ has a not too small probability of occurring. Once this is ensured, the uniformity test ensures that vertices have near-uniform amplitudes, and thus are meaningfully colored. Finally, the \textbf{consistency test} with respect to a given $\TCSP$ instance $\CTRS = (G,\{R_{e}\}_{e \in E})$ is:
\begin{itemize}
  \item[] $\CONSTEST_{\CTRS}\big(\ket{\BPSI^{(1)}},\ldots,\ket{\BPSI^{(\NUMP)}}\big) \equiv$
  \begin{enumerate}[nolistsep]
    \item For $i=1,\ldots,\NUMP$, measure the graph coloring state $\ket{\BPSI^{(i)}}$ in the standard basis to get outcome $(v_{i},j_{i})$.
    \item If there exist distinct $i,i' \in \{1,\ldots,\NUMP\}$ such that $v_{i}=v_{i'}$ but $j_{i} \neq j_{i'}$, then reject.
    \item If there exist distinct $i,i' \in \{1,\ldots,\NUMP\}$ such that $(v_{i},v_{i'}) \in E$ but $R_{(v_{i},v_{i'})}(j_{i},j_{i'})=0$, reject. 
    \item Accept.
  \end{enumerate}
\end{itemize}
The consistency test just checks that the states meaningfully encode a solution to the $\TCSP$ instance, by ensuring that each vertex has a unique color, and no edge is violated.  This test is only meaningful with honest encodings of the graph coloring state, and we can perform other tests (such as the conditional uniformity test) to rule out dishonest encodings.

Throughout, we will denote by $\REJ(\cdot)$ the rejection probability of a given test; e.g., $\REJ\big(\SWAPTEST\big(\ket{\BPSI},\ket{\BPHI}\big)\big)$ denotes the rejection probability of the swap test on the two quantum states $\ket{\BPSI}$ and $\ket{\BPHI}$.

\section{An Improvement on the Soundness Analysis of \cite{CD10}}
\label{sec:CD10-improved-soundness}

In this section, we give the details for our tight soundness analysis of the two-prover $\QMA$ protocol of Chen and Drucker \cite{CD10}. Specifically, we prove:

\begin{proposition*}[\propositionref{proposition:improved:CD10}, restated]
The $\NUMP$-prover $\QMA$ protocol for $\TCSP(N,M,K)$ given by \algorithmref{alg:CD10} has completeness $1-e^{-\Omega(\NUMP)}$ and soundness $1-\Omega\left(\frac{\NUMP^{2}}{N+\NUMP^{2}}\right)$, assuming $K \in O(1)$; thus, for $\NUMP \in \Omega(\log N)$ and $\NUMP \in O(\sqrt{N})$, the soundness gap is $\Omega(\NUMP^{2}N^{-1})$. Moreover, the analysis of the soundness of the protocol cannot be improved, in the sense of \remarkref{remark:CD10-soundness-tight} below.
\end{proposition*}

The proposition improves the status quo by giving a smooth trade-off between the number of provers $\NUMP$ and the soundness gap as a function of $\NUMP$, whereas the soundness analysis of \cite{CD10} only gave a soundness gap for $\NUMP \in \Theta(\sqrt{N})$.

\begin{algorithm}[h!]
\caption{Verifier of \cite{CD10}}
\label{alg:CD10}
\begin{flushleft}
\textbf{inputs:} a $\TCSP(N,M,K)$ instance $\CTRS = (G,\{R_{e}\}_{e \in E})$\\
\vspace{2mm}
\textbf{proofs:} $\NUMP$ unentangled graph coloring states $\ket{\BPSI^{(1)}},\ldots,\ket{\BPSI^{(\NUMP)}}$\\
\vspace{2mm}
\textbf{verifier:} draw $r \in \{1,2\}$ at random, and perform the $r$-th test below:
\begin{enumerate}[nolistsep]
   \item $\CONDUNITEST_{\frac{99}{100}\NUMP}\big(\ket{\BPSI^{(1)}},\ldots,\ket{\BPSI^{(\NUMP)}}\big)$
   \item $\CONSTEST_{\CTRS}\big(\ket{\BPSI^{(1)}},\ldots,\ket{\BPSI^{(\NUMP)}}\big)$
\end{enumerate}
\end{flushleft}  
\end{algorithm}

\begin{remark}[``Tightness'' of Our Analysis]
\label{remark:CD10-soundness-tight}
Consider a $\TCSP(N,M,K)$ instance $\CTRS = (G,\{R_{e}\}_{e \in E})$; suppose that $\CTRS$ is \emph{not} satisfiable, and suppose also that $\CTRS$ has constant soundness gap $\eta$. Hence, for any coloring $C \colon V \to \Sigma$, at least $\eta |E|$ of the edge constraints $\{R_{e}\}_{e \in E}$ are \emph{not} satisfied. So fix some coloring $C$.

Now suppose that the $\NUMP$ graph coloring states $\ket{\BPSI^{(1)}},\ldots,\ket{\BPSI^{(\NUMP)}}$ given to the verifier are all equal and indeed are a uniform superposition of all vertices with a unique color determined by $C$. If so, the test $\CONSTEST_{\CTRS}\big(\ket{\BPSI^{(1)}},\ldots,\ket{\BPSI^{(\NUMP)}}\big)$ rejects with probability $O(\frac{\NUMP^{2}}{N})$ by the Birthday Problem (indeed, we only have $\NUMP^{2}$ chances to see a particular edge in the constraint graph, and a constrained edge is seen with only probability $\Theta(N^{-1})$ because the graph is sparse). Thus, our analysis is ``tight'' in the sense that the assumptions we made could indeed really be the case, so one cannot hope to exhibit an even better soundness analysis that proves a soundness gap of $\omega(\NUMP^{2}/N)$.

Furthermore, if instead $\CTRS$ is satisfiable (and the verifier receives uniform and equal $\NUMP$ graph coloring states $\ket{\BPSI^{(1)}},\ldots,\ket{\BPSI^{(\NUMP)}}$ with a satisfying coloring), then completeness would be only $1 - e^{-\Theta(\NUMP)}$, due to the imperfect completeness of the conditional uniformity test.  This test has imperfect completeness due to Line \eqref{cut} of that test, that rejects whenever the number of 0's measured is below the threshold.  Due to natural variability, this can happen with non-zero probability even in the satisfiable case. Thus, we are forced to take $\NUMP \in \Omega(\log N)$ in order for there to be any inverse-polynomial soundness gap. (In other words, the protocol of \cite{CD10} has no soundness gap in the ``constant regime'' $\NUMP \in O(1)$; to breach the constant regime, it seems that one would have to strengthen the verifier with additional LOCC measurements to increase the soundness gap, or, at the very least, to endow the protocol with perfect completeness.)
\end{remark}

We now proceed to the proof of \propositionref{proposition:improved:CD10}, which follows closely the proof of Chen and Drucker \cite{CD10}. Throughout, we use notation for graph coloring states, which was introduced in \sectionref{sec:graph-states}.

Observe that the completeness in \propositionref{proposition:improved:CD10} follows exactly as in the analysis of
Chen and Drucker.  Thus, it remains to examine the soundness. Chen and Drucker \cite[Lemma 3]{CD10} gave sufficient conditions for an arbitrary graph coloring state $\ket{\BPSI}$ to be accepted by the uniformity test $\UNITEST(\ket{\BPSI})$ with constant probability. We first show how to use the generic lemmas of \sectionref{sec:graph-states} to prove the same result (and these same lemmas are used with very different parameters in our soundness analysis of the protocol of Blier and Tapp \cite{BT09} in \sectionref{sec:BT09-improved-soundness}). In particular, this next lemma says that, assuming the 0 coloring is measured with good probability, we can reject the dishonest case of when the provers assign too small amplitude to many vertices.

\begin{lemma}
Fix $\veps \in [0,1]$. If $p_{0}(\ket{\BPSI}) \geq \frac{1}{4K}$ and $\Abs{R_{\frac{1}{8KN}}\big(\ket{\BPSI}\big)} \geq \veps N$, then $\REJ\big(\UNITEST\big(\ket{\BPSI}\big)\big) \ge \frac{\veps^{2}}{64K}$.
\end{lemma}

\begin{proof}
For each $v \in R_{\frac{1}{8KN}}\big(\ket{\BPSI}\big)$, by invoking \lemmaref{lemma:gen-uniformity} with $j=0$, $c_{1}=4K$, and $c_{2}=8K$, we get that $\Abs{\gamma_{v}(0)}^{2} < \frac{1}{2N}$. In particular, we deduce that
\begin{align*}
\sum_{v=0}^{N-1} \Abs{ \Abs{\gamma_{v}(0)}^{2} - \dfrac{1}{N} }
\geq \sum_{v \in R_{\frac{1}{8KN}}(\ket{\BPSI})} \Abs{ \Abs{\gamma_{v}(0)}^{2} - \dfrac{1}{N} }
\geq \Abs{R_{\frac{1}{8KN}}(\ket{\BPSI})} \cdot \dfrac{1}{2N}
\geq \veps N \cdot \dfrac{1}{2N}
=    \dfrac{\veps}{2}
\enspace.
\end{align*}
Next, by invoking \lemmaref{lemma:vertex-and-Fourier} with $\ket{\gamma}=\ket{\gamma(0)}$, we get that the probability of measuring 0\ in the (only) register of $F_{N}\ket{\gamma(0)}$ is at most
\begin{equation*}
1 - \dfrac{1}{4} \left( \sum_{v=0}^{N-1} \Abs{ \Abs{\gamma_{v}(0)}^{2} - \dfrac{1}{N} } \right)^{2} \enspace.
\end{equation*}
We deduce that the probability of measuring 0\ in the (only) register of $F_{N}\ket{\gamma(0)}$ is at most $1-\frac{\veps^{2}}{16}$.  Thus, the probability of measuring $j=0$ and $v\ne 0$ is at least $\frac{\veps^2}{16}\cdot \frac{1}{4K}=\frac{\veps^2}{64K}$ as desired.
\end{proof} 

The above result shows that for graph coloring states with constant probability of measuring the 0 color, we reject with good probability if there are many vertices with small amplitudes.  In the case when 0 is not measured with such probability, nothing can be said.  Thus, Chen and Drucker argue that amongst the different graph coloring states, we can detect if very few of them have a good probability of measuring 0.  This can simply be done by measuring said states and comparing the number of zeroes measured and the relevant threshold value.  Thus, the remaining case to analyze is when there are many states with good probability of measuring 0, and each state has few vertices with small amplitude.  They give a reduction (with some loss in the constants) to a slightly simpler normal form of this case, which they then analyze.  We present a slightly better analysis of this normal form.

\begin{lemma}[modified {\cite[Lemma 4]{CD10}}]
Let $G=(V,E)$ be a $d$-regular graph (possibly with self-loops) with $N$ vertices, $M$ edges, and $d>1$. Let $\CTRS=(G,\{R_{e}\}_{e})$ be a $\TCSP$ on the graph $G$ with color alphabet $\Sigma$, and suppose that $\CTRS$ is $(1-\eta)$-unsatisfiable. Let $D_{1},\ldots,D_{\NUMP}$ be independent distributions on $V \times \Sigma$, where $(v_{i},c_{i})$ denotes the output of $D_{i}$.  

Suppose that for each $i \in \{1,\ldots,\NUMP\}$ there exists $S_{i} \subseteq V$ with $|S_{i}| \geq (1-\veps) N$ such that $v_{i}$ is uniformly distributed over $S_{i}$, and $\veps<\eta/20$. Then, when sampling $(v_i,c_i)$ from $D_i$ for all $i$, there is a probability of at least $\Omega_{\veps,d}(\frac{\NUMP^{2}}{N+\NUMP^{2}})$ such that there exists an $i<j$ with: either $e=(v_{i},v_{j})$ is an edge of $G$ and $R_{e}(c_{i},c_{j})=0$, or $v_{i}=v_{j}$ and $c_{i} \neq c_{j}$.
\end{lemma}

\begin{proof}
We follow the proof of Chen and Drucker \cite{CD10}. For $i,j \in \{1,\dots,\NUMP\}$, define $V_{i,j}$ to be an indicator for the event that either $e=(v_{i},v_{j})$ is an edge of $G$ and  $R_{e}(c_{i},c_{j})=0$, or $v_{i}=v_{j}$ and $c_{i} \neq c_{j}$. Denote $V=\sum_{i=1}^{\NUMP-1} \sum_{j=i+1}^{\NUMP} V_{i,j}$. Observe that the result follows from bounding $\Pr[V=0]$. To bound this probability, we use Cantelli's inequality (also known as the one-sided Chebyschev inequality, cf \cite{Ros84}): for a random variable $X$ and $a>0$, $\Pr[X\le \mathbb{E}[X]-a]\le\frac{\Var(X)}{\Var(X)+a^{2}}$. Thus, taking $X=V$ and $a=\mathbb{E}(V)$, and using the fact that $V$ is a non-negative random variable, we have
\begin{equation*}
     \Pr[V=0]
\leq \frac{\Var(V)}{\Var(V)+\mathbb{E}[V]^{2}}
=    1 - \frac{1}{\frac{\Var(V)}{\mathbb{E}[V]^{2}}+1}
\enspace.
\end{equation*}
The result will then follow from an upper bound on $\Var(V)$ and a lower bound on $\mathbb{E}[V]^{2}$.

We now invoke the following facts from the analysis of \cite{CD10}:
\begin{enumerate}[label=(\roman{*})] 
  \item $\mathbb{E}[V_{i,j}] \geq \veps/N$, and
  \item $\Var(V) = O_{\veps,d}(\NUMP^{2}/N+\NUMP^{3}/N^{2})$.
\end{enumerate}
Hence, the upper bound for $\Var(V)$ is already given. As for the lower bound on $\mathbb{E}[V]^{2}$: by linearity of expectation and (i) above, we see that $\mathbb{E}[V] = \binom{\NUMP}{2}\mathbb{E}[V_{i,j}] = \Omega_{\veps}(\NUMP^{2}/N)$; thus, $\mathbb{E}[V]^{2} \geq \Omega_{\veps}(\NUMP^{4}/N^{2})$. Therefore,
\begin{equation*}
     \dfrac{\Var(V)}{\mathbb{E}[V]^{2}}
\leq O_{\veps,d}\left(\frac{\NUMP^{2}/N+\NUMP^{3}/N^{2}}{\NUMP^{4}/N^{2}} \right)
\leq O_{\veps,d}\left(\frac{N+\NUMP}{\NUMP^{2}} \right)
\enspace.
\end{equation*}
Combining with the above, we conclude that
\begin{align*}
      \Pr \big[ V=0 \big]
\leq 1 - \frac{1}{\frac{\Var(V)}{\mathbb{E}[V]^{2}}+1}
\leq 1 - \frac{1}{O_{\veps,d}\left(\frac{N+\NUMP}{\NUMP^{2}}\right)+1}
\leq 1 - \Omega_{\veps,d}\left(\frac{\NUMP^{2}}{N+\NUMP^{2}}\right)
\enspace,
\end{align*}
where the big-$O$ and big-$\Omega$ notation hide constants depending on $\veps$ and $d$.  Thus, we obtain the desired lower bound on $\Pr[V>0]$.
\end{proof}

The above lemma, combined with the rest of Chen and Drucker's analysis, readily yield \propositionref{proposition:improved:CD10}.  That is, \cite{CD10} use the conditional uniformity test to rule out dishonest provers presenting malformed graph coloring states, and then use an analysis of the above type to analyze the case of dishonest provers presenting invalid colorings.  With the improved analysis, we can analyze the protocol in a larger parameter regime, giving the claim.

\section{An Improvement on the Soundness Analysis of \cite{BT09}}
\label{sec:BT09-improved-soundness}

In this section, we give the details for our tight soundness analysis of the two-prover $\QMA$ protocol of Blier and Tapp \cite{BT09}. Specifically, we prove:

\begin{proposition*}[\propositionref{proposition:improved:BT09}, restated]
The two-prover $\QMA$ protocol for $\TCSP(N,M,K)$ given in \algorithmref{alg:BT09} has (perfect completeness and) soundness $1-\Omega(N^{-2})$, assuming $K \in O(1)$. Moreover, the analysis of the soundness of the protocol cannot be improved, in the sense of \remarkref{remark:BT09-soundness-tight} below.
\end{proposition*}

\begin{algorithm}[h!]
\caption{Verifier of \cite{BT09}}
\label{alg:BT09}
\begin{flushleft}
\textbf{inputs:} a $\TCSP(N,M,K)$ instance $\CTRS = (G,\{R_{e}\}_{e \in E})$\\
\vspace{2mm}
\textbf{proofs:} two unentangled graph coloring states $\ket{\BPSI^{(1)}}$ and $\ket{\BPSI^{(2)}}$\\
\vspace{2mm}
\textbf{verifier:} draw $r \in \{1,2,3\}$ at random, and perform the $r$-th test below:
\begin{enumerate}[nolistsep]
   \item $\SWAPTEST\big(\ket{\BPSI^{(1)}},\ket{\BPSI^{(2)}}\big)$
   \item $\CONSTEST_{\CTRS}\big(\ket{\BPSI^{(1)}},\ket{\BPSI^{(2)}}\big)$
   \item $\UNITEST\big(\ket{\BPSI^{(1)}}\big) \,\wedge\, \UNITEST\big(\ket{\BPSI^{(2)}}\big)$
\end{enumerate}
\end{flushleft}  
\end{algorithm}

\begin{remark}[``Tightness'' of Our Analysis]
\label{remark:BT09-soundness-tight}
Consider a $\TCSP(N,M,K)$ instance $\CTRS = (G,\{R_{e}\}_{e \in E})$; suppose that $\CTRS$ is \emph{not} satisfiable, and suppose further that there exists a coloring of the vertices $C \colon V \to \Sigma$ for which there exists exactly one edge $(\tilde{v},\tilde{w}) \in E$ such that $R_{(\tilde{v},\tilde{w})}(C(\tilde{v}),C(\tilde{w}))=0$.

Now suppose that the two graph coloring states $\ket{\BPSI^{(1)}}$ and $\ket{\BPSI^{(2)}}$ given to the verifier are equal and that they indeed are a uniform superposition of all vertices, colored with $C$. If so, both the first test (i.e., the swap test) and the third test (i.e., the two uniformity tests) accept with probability $1$; however, the second test (i.e., the consistency test) accepts with probability that is exactly $1 - N^{-2}$.

In other words, our analysis is ``tight'' in the sense that the assumptions we made could indeed really be the case, thus implying that one cannot hope to exhibit an even better soundness analysis that proves a soundness of $1 - \omega(N^{-2})$.
\end{remark}

We now proceed to the proof of \propositionref{proposition:improved:BT09}, which we tackle in several lemmas, whose overall structure follows the approach taken by \cite{BT09}. Throughout, we use notation for graph coloring states introduced in \sectionref{sec:graph-states}. Also, given a $\TCSP$ instance $\CTRS$, $\COLCONSTEST_{\CTRS}\big(\ket{\BPSI^{(1)}},\ket{\BPSI^{(2)}}\big)$ denotes only the color consistency subtest of the test $\CONSTEST_{\CTRS}\big(\ket{\BPSI^{(1)}},\ket{\BPSI^{(2)}}\big)$ and $\EDGECONSTEST_{\CTRS}\big(\ket{\BPSI^{(1)}},\ket{\BPSI^{(2)}}\big)$ denotes only the edge consistency subtest of the test $\CONSTEST_{\CTRS}\big(\ket{\BPSI^{(1)}},\ket{\BPSI^{(2)}}\big)$.

First, we show that, as long as two graph coloring states $\ket{\BPSI^{(1)}}$ and $\ket{\BPSI^{(2)}}$ are ``close enough'' and the colors of the vertices are ``consistent enough'', then a definite color can be chosen for vertices with large enough amplitude. (Indeed, if vertices with large enough amplitude were to be colored very inconsistently, then we would be able to catch them, through the second test.)

\begin{lemma}[modified {\cite[Lemma 3.4]{BT09}}]
\label{lemma:S-has-unique-coloring}
Fix any $\TCSP$ instance $\CTRS$. Define
\begin{equation*}
\delta = \dfrac{1}{2 \cdot 1600^{2}K^{4}N^{2}} \enspace\text{ and }\enspace
\mu    = \dfrac{1}{1600^{2}K^{4}N^{2}} \enspace.
\end{equation*}
Suppose that:
\begin{itemize}
  \item[(i)] $\REJ\big(\SWAPTEST\big(\ket{\BPSI^{(1)}},\ket{\BPSI^{(2)}}\big)\big) \leq \delta$, and
  \item[(ii)] $\REJ\big(\COLCONSTEST_{\CTRS}\big(\ket{\BPSI^{(1)}},\ket{\BPSI^{(2)}}\big)\big) \leq \mu$.
\end{itemize}
Then, for every vertex $v \in S_{\frac{1}{8N}}\big(\ket{\BPSI^{(1)}}\big)$, there exists a (unique) $j \in \{0,\ldots,K-1\}$ such that $\abs{\beta_{v,j}^{(1)}}^{2} \geq \frac{100K-1}{100K}$. (And, similarly, for every vertex $v \in S_{\frac{1}{8N}}\big(\ket{\BPSI^{(2)}}\big)$, there exists a (unique) $j \in \{0,\ldots,K-1\}$ such that $\abs{\beta_{v,j}^{(2)}}^{2} \geq \frac{100K-1}{100K}$.)
\end{lemma}

\begin{proof}
First, if such a $j$ exists, it is unique, because $\frac{100K-1}{100K} > \frac{1}{2}$. Next suppose for the sake of contradiction that there exists some vertex $\tilde{v} \in S_{\frac{1}{8N}}\big(\ket{\BPSI^{(1)}}\big)$ for which there is no such $j$, so that there exist distinct $j_{1},j_{2} \in \{0,\ldots,K-1\}$ such that $\abs{\beta_{\tilde{v},j_{1}}^{(1)}}^{2},\abs{\beta_{\tilde{v},j_{2}}^{(1)}}^{2} > \frac{1}{100K^{2}}$.\footnote{Indeed, from $\abs{\beta_{\tilde{v},0}^{(1)}}^{2},\ldots,\abs{\beta_{\tilde{v},K-1}^{(1)}}^{2} < \frac{100K-1}{100K}$ and $\sum_{j=0}^{K-1} \abs{\beta_{\tilde{v},j}^{(1)}}^{2} = 1$, we deduce that there exists some $j_{1} \in \{0,\ldots,K-1\}$ such that $\abs{\beta_{\tilde{v},j_{1}}^{(1)}}^{2} \geq \frac{1}{K}$ and, from $\abs{\beta_{\tilde{v},0}^{(1)}}^{2},\ldots,\abs{\beta_{\tilde{v},K-1}^{(1)}}^{2} < \frac{100K-1}{100K}$ and $\sum_{j \neq j_{1}} \abs{\beta_{\tilde{v},j}^{(1)}}^{2} > \frac{1}{100K}$, we deduce that there exists some $j_{2} \in \{0,\ldots,K-1\}-\{j_{1}\}$ such that $\abs{\beta_{\tilde{v},j_{2}}^{(1)}}^{2} \geq \frac{1}{100K(K-1)} > \frac{1}{100K^{2}}$. Overall, $\abs{\beta_{\tilde{v},j_{1}}^{(1)}}^{2},\abs{\beta_{\tilde{v},j_{2}}^{(1)}}^{2} > \frac{1}{100K^{2}}$.} Then, the probability that the color-consistency test rejects the two graph coloring states $\ket{\BPSI^{(1)}}$ and $\ket{\BPSI^{(2)}}$, is
\begin{align*}
      &\sum_{v=0}^{N-1} \sum_{j=0}^{K-1} \sum_{j' \neq j} \Abs{\alpha_{v}^{(1)} \beta_{v,j}^{(1)}}^{2} \cdot \Abs{\alpha_{v}^{(2)} \beta_{v,j'}^{(2)}}^{2} \\
\geq\,& \sum_{v=0}^{N-1} \sum_{j=0}^{K-1} \sum_{j' \neq j} \Abs{\alpha_{v}^{(1)} \beta_{v,j}^{(1)}}^{2} \cdot \left( \Abs{\alpha_{v}^{(1)} \beta_{v,j'}^{(1)}}^{2} - \sqrt{2\delta} \right) \hspace{1cm} \text{(by \lemmaref{lemma:swap-test-and-measurements})} \\
\geq\,& \sum_{v \in S_{\frac{1}{8N}}(\ket{\BPSI^{(1)}})} \sum_{j=0}^{K-1} \sum_{j' \neq j} \dfrac{\Abs{\beta_{v,j}^{(1)}}^{2}}{8N} \cdot \left( \dfrac{\Abs{\beta_{v,j'}^{(1)}}^{2}}{8N} - \sqrt{2\delta} \right) \\
\geq\,& \sum_{j=0}^{K-1} \sum_{j' \neq j} \dfrac{\Abs{\beta_{\tilde{v},j}^{(1)}}^{2}}{8N} \cdot \left( \dfrac{\Abs{\beta_{\tilde{v},j'}^{(1)}}^{2}}{8N} - \sqrt{2\delta} \right) \\
\geq\,& \dfrac{\Abs{\beta_{\tilde{v},j_{1}}^{(1)}}^{2}}{8N} \cdot \left( \dfrac{\Abs{\beta_{\tilde{v},j_{2}}^{(1)}}^{2}}{8N} - \sqrt{2\delta} \right) \\
>\,   & \dfrac{1}{800K^{2}N} \cdot \left( \dfrac{1}{800K^{2}N} - \sqrt{2\delta} \right) \hspace{1cm} \text{(observe the strict inequality)}\\
\geq\,& \dfrac{1}{800K^{2}N} \cdot \left( \dfrac{1}{800K^{2}N} - \sqrt{2 \cdot \dfrac{1}{2\cdot 1600^{2}K^{4}N^{2}}} \right) \\
\geq\,& \dfrac{1}{1600^{2}K^{4}N^{2}}=\mu
\enspace,
\end{align*}
which contradicts the assumption that $\REJ\big(\COLCONSTEST_{\CTRS}\big(\ket{\BPSI^{(1)}},\ket{\BPSI^{(2)}}\big)\big) \leq \mu$. An analogous argument holds for $\ket{\BPSI^{(2)}}$.
\end{proof}

Next, we show that, under the same assumptions as \lemmaref{lemma:S-has-unique-coloring}, the probability of measuring $j$ in the color register of $(I_{N} \otimes F_{K})\ket{\BPSI^{(1)}}$ is at least $\frac{1}{4K}$ for every color $j \in \{0,\ldots,K-1\}$.

\begin{lemma}[modified {\cite[Lemma 3.5]{BT09}}]
\label{lemma:get-0-color-often}
Fix any $\TCSP$ instance $\CTRS$. Define
\begin{equation*}
\delta = \dfrac{1}{2 \cdot 1600^{2}K^{4}N^{2}} \enspace\text{ and }\enspace
\mu    = \dfrac{1}{1600^{2}K^{4}N^{2}} \enspace.
\end{equation*}
Suppose that:
\begin{itemize}
  \item[(i)] $\REJ\big(\SWAPTEST\big(\ket{\BPSI^{(1)}},\ket{\BPSI^{(2)}}\big)\big) \leq \delta$, and
  \item[(ii)] $\REJ\big(\COLCONSTEST_{\CTRS}\big(\ket{\BPSI^{(1)}},\ket{\BPSI^{(2)}}\big)\big) \leq \mu$.
\end{itemize}
Then, $p_{j}\big(\ket{\BPSI^{(1)}}\big) \geq \frac{1}{4K}$ for every $j \in \{0,\ldots,K-1\}$. (And, similarly, $p_{j}\big(\ket{\BPSI^{(2)}}\big) \geq \frac{1}{4K}$ for every $j \in \{0,\ldots,K-1\}$.)
\end{lemma}

\begin{proof}
Suppose that the vertex register of the graph coloring state $\ket{\BPSI^{(1)}}$ is measured and that the outcome is some vertex $v \in \{0,\ldots,N-1\}$. If $v \in S_{\frac{1}{8N}}\big(\ket{\BPSI^{(1)}}\big)$, from \lemmaref{lemma:S-has-unique-coloring} we deduce that there exists a (unique) color $\tilde{j} \in \{0,\ldots,K-1\}$ such that $\abs{\beta_{v,\tilde{j}}^{(1)}}^{2} \geq \frac{100K-1}{100K}$; in particular, we also deduce that $\sum_{j \neq \tilde{j}} \abs{\beta_{v,j}^{(1)}}^{2} < \frac{1}{100K}$. Therefore, (conditioned on getting outcome $v$ in the vertex register) the probability of measuring $j$ in the color register of $(I_{N} \otimes F_{K})\ket{\BPSI^{(1)}}$ is
\begin{align*}
      & \dfrac{1}{K} \Abs{\sum_{j=1}^{K-1} \beta_{v,j}^{(1)}e^{\frac{2 \pi \sqrt{-1} jv}{K}}}^{2} \\
\geq\,& \dfrac{1}{K} \Abs{\Abs{\beta_{v,\tilde{j}}^{(1)}e^{\frac{2 \pi \sqrt{-1} \tilde{j}v}{K}}} - \Abs{\sum_{j \neq \tilde{j}} \beta_{v,j}^{(1)}e^{\frac{2 \pi \sqrt{-1} jv}{K}}}}^{2} \\
\geq\,& \dfrac{1}{K} \Abs{\Abs{\beta_{v,\tilde{j}}^{(1)}e^{\frac{2 \pi \sqrt{-1} \tilde{j}v}{K}}} - \sqrt{K\sum_{j \neq \tilde{j}} \Abs{\beta_{v,j}^{(1)}e^{\frac{2 \pi \sqrt{-1} jv}{K}}}^{2}}}^{2} \hspace{1cm}\text{(by Cauchy--Schwarz)} \\
=   \,& \dfrac{1}{K} \Abs{\Abs{\beta_{v,\tilde{j}}^{(1)}} - \sqrt{K\sum_{j \neq \tilde{j}} \Abs{\beta_{v,j}^{(1)}}^{2}}}^{2} \\
\geq\,& \dfrac{1}{K} \Abs{\sqrt{\dfrac{100K-1}{100K}} - \sqrt{K\dfrac{1}{100K}}}^{2} \\
\geq\,& \dfrac{1}{K} \left(1 - \dfrac{1}{100K} + \dfrac{1}{100} - \dfrac{1}{5} \cdot \dfrac{100K-1}{100K}\right) \\
=   \,& \dfrac{4}{5K}
\enspace.
\end{align*}
Now observe that $S_{\frac{1}{8N}}\big(\ket{\BPSI^{(1)}}\big)$ cannot be empty, for otherwise $\sum_{v=0}^{N-1} \abs{\alpha_{v}^{(1)}}^{2} < N \cdot \frac{1}{8N} < 1$. Hence, there is at least one vertex $\tilde{v}$ in $S_{\frac{1}{8N}}(\ket{\BPSI^{(1)}})$. Thus, the probability of measuring $j$ (with no conditioning) in the color register of $(I_{N} \otimes F_{K})\ket{\BPSI^{(1)}}$ is
\begin{align*}
p_{j}\big(\ket{\BPSI^{(1)}}\big) =
      & \sum_{v=0}^{N-1} \Abs{\alpha_{v}^{(1)}}^{2} \dfrac{1}{K} \Abs{\sum_{j=0}^{K-1} \beta_{v,j}^{(1)}e^{\frac{2 \pi \sqrt{-1} jv}{K}}}^{2} \\
\geq\,& \sum_{v \in S_{\frac{1}{8N}}(\ket{\BPSI^{(1)}})} \Abs{\alpha_{v}^{(1)}}^{2} \dfrac{1}{K} \Abs{\sum_{j=0}^{K-1} \beta_{v,j}^{(1)}e^{\frac{2 \pi \sqrt{-1} jv}{K}}}^{2} \\
\geq\,& \dfrac{4}{5K} \sum_{v \in S_{\frac{1}{8N}}(\ket{\BPSI^{(1)}})} \Abs{\alpha_{v}^{(1)}}^{2} \\
\geq\,& \dfrac{4}{5K} \left( 1 - (N-1) \cdot \dfrac{1}{8N} \right) \\
\geq\,& \dfrac{4}{5K} \cdot \dfrac{7}{8} \\
\geq\,& \dfrac{1}{4K}
\enspace,
\end{align*}
as desired. An analogous argument holds for $\ket{\BPSI^{(2)}}$.
\end{proof}

Next we prove that, under the same assumptions of \lemmaref{lemma:S-has-unique-coloring} and \lemmaref{lemma:get-0-color-often}, if we further require that the uniform test does not reject with high probability, then we can be sure that all the vertices have a somewhat large amplitude.

\begin{lemma}[modified {\cite[Lemma 3.7]{BT09}}]
\label{lemma:all-vertices}
Fix any $\TCSP$ instance $\CTRS$. Define
\begin{equation*}
\delta = \dfrac{1}{2 \cdot 1600^{2}K^{4}N^{2}} \enspace\text{ and }\enspace
\mu    = \dfrac{1}{1600^{2}K^{4}N^{2}} \enspace\text{ and }\enspace
\nu    = \dfrac{1}{64KN^{2}}
\enspace.
\end{equation*}
Suppose that:
\begin{itemize}
  \item[(i)] $\REJ\big(\SWAPTEST\big(\ket{\BPSI^{(1)}},\ket{\BPSI^{(2)}}\big)\big) \leq \delta$, 
  \item[(ii)] $\REJ\big(\COLCONSTEST_{\CTRS}\big(\ket{\BPSI^{(1)}},\ket{\BPSI^{(2)}}\big)\big) \leq \mu$, and
  \item[(iii)] $\REJ\big(\UNITEST\big(\ket{\BPSI^{(1)}}\big)\big) \leq \nu$.
\end{itemize}
Then, $V(G)=S_{\frac{1}{8KN}}\big(\ket{\BPSI^{(1)}}\big)$, that is, for all $v\in\{0,\ldots,N-1\}$, $\abs{\alpha_{v}^{(1)}}^2\ge \frac{1}{8KN}$. (And, similarly, $V(G)=S_{\frac{1}{8KN}}\big(\ket{\BPSI^{(2)}}\big)$ under the alternative assumption $\REJ\big(\UNITEST\big(\ket{\BPSI^{(2)}}\big)\big) \leq \nu$ instead.)
\end{lemma}

\begin{proof}
Recall that:
\begin{itemize}
  \item $p_{j}\big(\ket{\BPSI^{(1)}}\big)$ is the probability of measuring $j$ in the color register of $(I_{N} \otimes F_{K})\ket{\BPSI^{(1)}}$, and
  \item $\ket{\gamma(j)^{(1)}} = \sum_{v=0}^{N-1} \gamma_{v}(j)^{(1)} \ket{v}$ is the reduced quantum state obtained when we measure $j$ in the color register of $(I_{N} \otimes F_{K})\ket{\BPSI^{(1)}}$.
\end{itemize}
By invoking \lemmaref{lemma:vertex-and-Fourier} with $\ket{\gamma}=\ket{\gamma(j)^{(1)}}$, the probability of measuring $v$ in the vertex register of $F_{N}\ket{\gamma(j)^{(1)}}$ is at most
\begin{equation*}
1 - \dfrac{1}{4} \left( \sum_{v=0}^{N-1} \Abs{ \Abs{\gamma_{v}(j)^{(1)}}^{2} - \dfrac{1}{N} } \right)^{2} \enspace.
\end{equation*}
Also, by \lemmaref{lemma:get-0-color-often}, $p_{j}\big(\ket{\BPSI^{(1)}}\big) \geq \frac{1}{4K}$ for every $j \in \{0,\ldots,K-1\}$.

Suppose now by way of contradiction that there exists some vertex $\tilde{v} \in R_{\frac{1}{8KN}}\big(\ket{\BPSI^{(1)}}\big)$, so that $\abs{\alpha_{v}^{(1)}}^{2} < \frac{1}{8KN}$. We can now invoke \lemmaref{lemma:gen-uniformity} with $c_{1}=4K$ and $c_{2}=8K$ to get that $\abs{\gamma_{\tilde{v}}(j)^{(1)}}^{2} < \frac{1}{2N}$. Therefore,
\begin{align*}
     \sum_{v=0}^{N-1} \Abs{ \Abs{\gamma_{v}(j)^{(1)}}^{2} - \dfrac{1}{N} }
\geq \Abs{ \Abs{\gamma_{\tilde{v}}(j)^{(1)}}^{2} - \dfrac{1}{N} }
>    \dfrac{1}{2N}
\enspace,
\end{align*}
and we obtain that the probability of measuring $v$ in the vertex register of $F_{N}\ket{\gamma(j)^{(1)}}$ is less than $1-\frac{1}{16N^{2}}$. Thus, the probability of measuring $j$ in the color register but not measuring $v$ in the vertex register of $(F_{N} \otimes F_{K})\ket{\BPSI^{(1)}}$ is greater than
\begin{equation*}
\dfrac{1}{4K} \cdot \dfrac{1}{16N^{2}} = \dfrac{1}{64KN^{2}} = \nu \enspace.
\end{equation*}
Taking $j=0$ and $v=0$, this contradicts the assumption that $\REJ\big(\UNITEST\big(\ket{\BPSI^{(1)}}\big)\big) \leq \nu$. An analogous argument holds for $\ket{\BPSI^{(2)}}$.
\end{proof}

Finally, we can now lower bound the soundness of the protocol:

\begin{lemma}
\label{lemma:soundness-conclusion}
Define
\begin{equation*}
\delta = \dfrac{1}{2 \cdot 1600^{2}K^{4}N^{2}} \enspace\text{ and }\enspace
\mu    = \dfrac{1}{1600^{2}K^{4}N^{2}} \enspace\text{ and }\enspace
\nu    = \dfrac{1}{64KN^{2}} \enspace\text{ and }\enspace
\xi    = \dfrac{(100K-1)^{2}}{2 \cdot 800^{2}K^{4}N^{2}}
\enspace.
\end{equation*}
and
\begin{equation*}
s = \dfrac{1}{3} \min\left\{ \delta \,,\, \mu \,,\, \nu \,,\, \xi \right\} \enspace.
\end{equation*}
Then the overall probability of rejecting an unsatisfiable graph $G$ is greater than $s$.
\end{lemma}

\begin{proof}
If any of
\begin{itemize}
  \item[(i)] $\REJ\big(\SWAPTEST\big(\ket{\BPSI^{(1)}},\ket{\BPSI^{(2)}}\big)\big) \leq \delta$, 
  \item[(ii)] $\REJ\big(\COLCONSTEST_{\CTRS}\big(\ket{\BPSI^{(1)}},\ket{\BPSI^{(2)}}\big)\big) \leq \mu$, and
  \item[(iii)] $\REJ\big(\UNITEST\big(\ket{\BPSI^{(1)}}\big)\big) \leq \nu$ and $\REJ\big(\UNITEST\big(\ket{\BPSI^{(2)}}\big)\big) \leq \nu$
\end{itemize}
does not hold, then we are done. So suppose that (i)--(iii) hold. Define a coloring $C \colon V(G) \to \Sigma$ of the graph $G$ by the rule
\begin{equation*}
C(v) := \arg\max_{j \in \{0,\ldots,K-1\}} \Abs{\beta_{v,j}^{(1)}}^{2} \enspace,
\end{equation*}
for every vertex $v$. By \lemmaref{lemma:S-has-unique-coloring}, the coloring $C$ is well-defined (i.e., is unique).

Let $U(G) \subseteq E(G)$ be the set of unsatisfied edges in $G$ by the coloring $C$. Since $G$ is unsatisfiable, $\abs{U(G)} \geq 1$. Therefore, $\REJ\big(\EDGECONSTEST_{\CTRS}\big(\ket{\BPSI^{(1)}},\ket{\BPSI^{(2)}}\big)\big)$, which is the probability that the edge-consistency subtest rejects $\ket{\BPSI^{(1)}}$ and $\ket{\BPSI^{(2)}}$, is
\begin{align*}
    \,& \sum_{(v,w) \in U(G)} \Abs{\alpha_{v}^{(1)}\beta_{v,C(v)}^{(1)}}^{2} \cdot \Abs{\alpha_{w}^{(2)}\beta_{w,C(w)}^{(2)}}^{2} \\
\geq\,& \sum_{(v,w) \in U(G)} \left( \dfrac{1}{8KN} \cdot \dfrac{100K-1}{100K} \right) \cdot \left( \dfrac{1}{8KN} \cdot \dfrac{100K-1}{100K} \right) \\
=   \,& \Abs{U(G)} \cdot \dfrac{(100K-1)^{2}}{800^{2}K^{4}N^{2}} \\
\geq\,& 1 \cdot \dfrac{(100K-1)^{2}}{800^{2}K^{4}N^{2}} \\
>   \,& s
\enspace.
\end{align*}
This concludes the proof of the lemma, as well as the proof of \propositionref{proposition:improved:BT09}.
\end{proof}

\section*{Acknowledgements}
\label{sec:acks}

The authors would like to thank Scott Aaronson for his great lectures in quantum complexity theory and his suggestions while working on this note.

\bibliographystyle{alpha}
\bibliography{chiesa-forbes}

\newcommand{\etalchar}[1]{$^{#1}$}
\begin{thebibliography}{BCWdW01}

\bibitem[ABD{\etalchar{+}}09]{ABDFS09}
Scott Aaronson, Salman Beigi, Andrew Drucker, Bill Fefferman, and Peter Shor.
\newblock The power of unentanglement.
\newblock {\em Theory of Computing}, 5(1):1--42, 2009.
\newblock Earlier version appeared in CCC~'08.
  \href{http://arxiv.org/abs/0804.0802}{arXiv:0804.0802v2}.

\bibitem[AN02]{AN02}
Dorit Aharonov and Tomer Naveh.
\newblock Quantum $\mathrm{NP}$ - a survey.
\newblock Technical Report quant-ph/0210077, October 2002.
\newblock
  \href{http://arxiv.org/abs/quant-ph/0210077v1}{arXiv:quant-ph/0210077v1}.

\bibitem[BaCJ10]{BCY10}
Fernando Brand\~{a}o, Matthias Christiandl, and Yard Jon.
\newblock Faithful squashed entanglement.
\newblock {\em ArXiv e-prints}, October 2010.
\newblock \href{http://arxiv.org/abs/1010.1750v1}{arXiv:1010.1750v1}.

\bibitem[BaH12]{BrandaoH12}
Fernando Brand\~{a}o and Aram~W. Harrow.
\newblock Quantum de {F}inetti theorems under local measurements with
  applications, 2012.
\newblock arXiv:quant-ph/1210.6367.

\bibitem[BBD{\etalchar{+}}97]{BarencoBDEJM97}
Adriano Barenco, Andr{\'e} Berthiaume, David Deutsch, Artur Ekert, Richard
  Jozsa, and Chiara Macchiavello.
\newblock Stabilization of quantum computations by symmetrization.
\newblock {\em SIAM Journal on Computing}, 26(5):1541--1557, 1997.

\bibitem[BCWdW01]{BCWdW01}
Harry Buhrman, Richard Cleve, John Watrous, and Ronald de~Wolf.
\newblock Quantum fingerprinting.
\newblock {\em Physical Review Letters}, 87(16):167902, September 2001.
\newblock
  \href{http://arxiv.org/abs/quant-ph/0102001v1}{arXiv:quant-ph/0102001v1}.

\bibitem[Bei10]{Bei10}
Salman Beigi.
\newblock $\mathrm{NP}$ vs. $\mathrm{QMA}_{\log}(2)$.
\newblock {\em Quantum Information and Computation}, 54(1\&2):0141--0151, 2010.
\newblock \href{http://arxiv.org/abs/0810.5109v2}{arXiv:0810.5109v2}.

\bibitem[BSS08]{BS05}
Eli Ben-Sasson and Madhu Sudan.
\newblock Short {PCP}s with polylog query complexity.
\newblock {\em SIAM Journal on Computing}, 38(2):551--607, 2008.

\bibitem[BT09]{BT09}
Hugue Blier and Alain Tapp.
\newblock All languages in {NP} have very short quantum proofs.
\newblock In {\em ICQNM '09: Proceedings of the 3rd International Conference on
  Quantum, Nano and Micro Technologies}, pages 34--37, Los Alamitos, CA, USA,
  2009. IEEE Computer Society.
\newblock \href{http://arxiv.org/abs/0709.0738v1}{ arXiv:0709.0738v1}. Revised
  version: \href{http://arxiv.org/abs/0709.0738v2}{arXiv:0709.0738v2}.

\bibitem[CD10]{CD10}
Jing Chen and Andrew Drucker.
\newblock Short multi-prover quantum proofs for {SAT} without entangled
  measurements.
\newblock {\em ArXiv e-prints}, November 2010.
\newblock \href{http://arxiv.org/abs/1011.0716v2}{arXiv:1011.0716v2}.

\bibitem[Din07]{Din07}
Irit Dinur.
\newblock The {PCP} theorem by gap amplification.
\newblock {\em Journal of the ACM}, 54, June 2007.
\newblock Earlier version appeared in STOC~'06.

\bibitem[GS89]{GS89}
Yuri Gurevich and Saharon Shelah.
\newblock Nearly linear time.
\newblock In Albert~R. Meyer and Michael~A. Taitslin, editors, {\em Logic at
  Botik'89: Symposium on Logical Foundations of Computer Science}, pages
  108--118. Springer-Verlag New York, Inc., New York, NY, USA, 1989.

\bibitem[HM10]{HM10}
Aram Harrow and Ashley Montanaro.
\newblock An efficient test for product states, with applications to quantum
  merlin-arthur games.
\newblock In {\em FOCS '10: Proceedings of the 51st Annual IEEE Symposium on
  Foundations of Computer Science}, pages 633--642, Washington, DC, USA, 2010.
  IEEE Computer Society.
\newblock \href{http://arxiv.org/abs/1001.0017v3}{arXiv:1001.0017v3}.

\bibitem[IP01]{ImpagliazzoR01}
Russel Impagliazzo and Ramamohan Paturi.
\newblock On the complexity of k-{SAT}.
\newblock {\em Journal of Computer and System Sciences}, 62(2):367--375, March
  2001.

\bibitem[Kit99]{Kit99}
Alexei~Yu. Kitaev.
\newblock Quantum $\mathrm{NP}$.
\newblock Talk at AQIP~'99: Second Workshop on Algorithms in Quantum
  Information Processing, 1999.

\bibitem[KMY09]{KMY03}
Hirotada Kobayashi, Keiji Matsumoto, and Tomoyuki Yamakami.
\newblock Quantum {M}erlin-{A}rthur proof systems: Are multiple {M}erlins more
  helpful to {A}rthur?
\newblock In {\em Chicago Journal of Theorerical Computer Science}, volume~3,
  pages 1--18. 2009.
\newblock
  \href{http://arxiv.org/abs/quant-ph/0306051v2}{arXiv:quant-ph/0306051v2}.

\bibitem[Kni96]{Kni96}
Emanuel~H. Knill.
\newblock Quantum randomness and nondeterminism.
\newblock Technical Report quant-ph/9610012, Oct 1996.
\newblock
  \href{http://arxiv.org/abs/quant-ph/9610012v1}{arXiv:quant-ph/9610012v1}.

\bibitem[KSTW01]{KSTW01}
Sanjeev Khanna, Madhu Sudan, Luca Trevisan, and David~P. Williamson.
\newblock The approximability of constraint satisfaction problems.
\newblock {\em SIAM Journal on Computing}, 30:1863--1920, December 2001.

\bibitem[LCV07]{LCV07}
Yi-Kai Liu, Matthias Christandl, and Frank Verstraete.
\newblock Quantum computational complexity of the $n$-representability problem:
  $\mathrm{QMA}$ complete.
\newblock {\em Physical Review Letters}, 98(11):110503, March 2007.
\newblock
  \href{http://arxiv.org/abs/quant-ph/0609125v1}{arXiv:quant-ph/0609125v1}.

\bibitem[LGNN12]{GallNN12}
Fran\c{c}ois Le~Gall, Shota Nakagawa, and Harumichi Nishimura.
\newblock On {QMA} protocols with two short quantum proofs.
\newblock {\em Quantum Information and Computation}, 12(7-8):589--600, July
  2012.

\bibitem[NC00]{NC00}
Michael~A. Nielsen and Isaac~L. Chuang.
\newblock {\em Quantum Computation and Quantum Information}.
\newblock Cambridge University Press, New York, NY, USA, 2000.

\bibitem[Pap94]{Pap94}
Christos~H. Papadimitriou.
\newblock {\em {Computational Complexity}}.
\newblock Addison-Wesley, Reading, MA, USA, 1994.

\bibitem[Per12]{Pereszlenyi12}
Attila Pereszl{\'e}nyi.
\newblock Multi-prover quantum {M}erlin-{A}rthur proof systems with small gap,
  2012.
\newblock arXiv:quant-ph/1205.2761.

\bibitem[PY91]{PY91}
Christos~H. Papadimitriou and Mihalis Yannakakis.
\newblock Optimization, approximation, and complexity classes.
\newblock {\em Journal of Computer and System Sciences}, 43(3):425--440, 1991.
\newblock Earlier version appeared in STOC~'88.

\bibitem[Ros84]{Ros84}
Sheldon Ross.
\newblock {\em A first course in probability}.
\newblock Macmillan Co., New York, second edition, 1984.

\bibitem[Wat00]{Wat00}
John Watrous.
\newblock Succinct quantum proofs for properties of finite groups.
\newblock In {\em FOCS '00: Proceedings of the 41st Annual IEEE Symposium on
  Foundations of Computer Science}, pages 537--546, Washington, DC, USA, 2000.
  IEEE Computer Society.
\newblock \href{http://arxiv.org/abs/cs/0009002v1}{arXiv:cs/0009002v1}.

\end{thebibliography}

\end{document}